\documentclass[a4paper]{amsart}

\usepackage[varg]{pxfonts}
\usepackage{amssymb}
\usepackage[only,llbracket,rrbracket]{stmaryrd}

\usepackage{amsthm}
\usepackage{hyperref}
\usepackage{aliascnt}
\usepackage{amsaddr}

\DeclareSymbolFont{CMSY}{OMS}{cmsy}{m}{n}
\DeclareSymbolFont{MSA}{U}{msa}{m}{n}
\DeclareSymbolFont{MSB}{U}{msb}{m}{n}
\DeclareMathSymbol{\exists}{\mathord}{CMSY}{"39}
\DeclareMathSymbol{\vdash}{\mathrel}{CMSY}{"60}
\DeclareMathSymbol{\nvdash}{\mathrel}{MSB}{"30}
\DeclareMathSymbol{\vDash}{\mathrel}{MSA}{"0F}
\DeclareMathSymbol{\nvDash}{\mathrel}{MSB}{"32}
\DeclareMathSymbol{\Vdash}{\mathrel}{MSA}{"0D}
\DeclareMathSymbol{\nVdash}{\mathrel}{MSB}{"31}
\DeclareMathSymbol{\dashv}{\mathrel}{CMSY}{"61}

\usepackage{tikz}
\usetikzlibrary{arrows,calc}
\definecolor{darkgreen}{rgb}{0.2,0.8,0.2}

\usepackage{bussproofs}
\EnableBpAbbreviations

\newcommand{\Sets}{\mathbf{Sets}}
\newcommand{\DD}{\mathbf{S}}
\newcommand{\2}{\mathbf{2}}
\newcommand{\bool}{\2}
\newcommand{\C}{\mathcal{C}}
\newcommand{\pw}{\mathcal{P}}
\newcommand{\pwfin}{\mathcal{P}_\mathrm{fin}}
\newcommand{\Match}[2]{\mathrm{Match}_{{#1}, {#2}}}
\newcommand{\A}{\mathcal{A}}
\newcommand{\true}{1}

\newcommand{\cmp}{\mathbin{\circ}}
\newcommand{\rest}[2]{{#1}{\mid}_{#2}}
\newcommand{\natjoin}{\mathop{\bowtie}}
\newcommand{\Simp}{\mathbf{Simp}}
\newcommand{\ndSimp}{\mathbf{ndSimp}}
\newcommand{\op}{{\mathrm{op}}}
\renewcommand{\L}{\mathcal{L}}
\newcommand{\TT}{\mathbb{T}}
\newcommand{\sPsh}{\mathbf{sPsh}}
\newcommand{\incto}{\hookrightarrow}
\newcommand{\Sub}{\mathrm{Sub}}

\newcommand{\Scott}[1]{\llbracket{#1}\rrbracket}
\newcommand{\monoto}{\rightarrowtail}
\newcommand{\ldot}{\mathpunct{.}}

\newcommand{\MM}[2][]{\mathbb{M}_{#1}(#2)}
\newcommand{\ZZ}{\mathbb{Z}}
\newcommand{\NN}{\mathbb{N}}
\newcommand{\Filt}{\mathrm{Filt}}
\newcommand{\upset}[1]{\mathop{\uparrow}{#1}}
\newcommand{\FiltMM}[2][]{\mathbb{FM}_{#1}(#2)}

\newcommand{\NewTheorem}[3]{%
	\newaliascnt{#1}{#2}
	\newtheorem{#1}[#1]{#3}
	\aliascntresetthe{#1}
	\expandafter\def\csname #1autorefname\endcsname{#3}
}
\newtheorem{theorem}{Theorem}
\NewTheorem{lemma}{theorem}{Lemma}
\theoremstyle{definition}
\NewTheorem{definition}{theorem}{Definition}
\NewTheorem{example}{theorem}{Example}
\NewTheorem{fact}{theorem}{Fact}

\usepackage[hcentering,height=9in,width=6in,headsep=0.5in,foot=0.5in]{geometry}
\allowdisplaybreaks[4]
\title{%
Logic of Local Inference for\\
Contextuality in Quantum Physics and Beyond%
}
\author{%
Kohei Kishida%
}
\thanks{This work has been supported by the grant FA9550-12-1-0136 of the U.S. AFOSR}
\address{University of Oxford}
\email{Kohei.Kishida@cs.ox.ac.uk}
\date{}

\begin{document}

\maketitle

\begin{abstract}
Contextuality in quantum physics provides a key resource for quantum information and computation.
The topological approach in \cite{AB11,ABKLM15} characterizes contextuality as ``global inconsistency'' coupled with ``local consistency'', revealing it to be a phenomenon also found in many other fields.
This has yielded a logical method of detecting and proving the ``global inconsistency'' part of contextuality.
Our goal is to capture the other, ``local consistency'' part, which requires a novel approach to logic that is sensitive to the topology of contexts.
To achieve this, we formulate a logic of local inference by using context-sensitive theories and models in regular categories.
This provides a uniform framework for local consistency, and lays a foundation for high-level methods of detecting, proving, and moreover using contextuality as computational resource.
\end{abstract}

\section{Introduction}\label{sec:intro}

Quantum physics provides quantum computing with immense advantage over classical computing.
Among its non-classical properties, non-locality is known to be a basis for quantum communication (see \cite{NC11}).
It is in fact a special case of a more general property called \emph{contextuality}, which recent studies \cite{rau13,HWVE14} suggest is an essential source of the computational power of quantum computers.
This motivates the search for structural, higher-level expressions of contextuality that are independent of the formalism of quantum mechanics.

One conception of contextuality that originated in \cite{IB98} exploits the structure of \emph{presheaf}:
As was shown in \cite{IB98}, a certain type of contextuality of a quantum system amounts to the absence of global sections from presheaves modelling the behaviors of the system.
The recent, ``sheaf-theoretic'' approach \cite{AB11} expands this insight by viewing contextuality in more general terms, as a matter of \emph{topology} in data of measurements and outcomes:
A wider range of contextuality is then characterized as the ``global inconsistency'' of the ``locally consistent''.
This has on the one hand
shown that contextual phenomena can be found in various other fields such as relational database theory (see \cite{abr13}),
and on the other hand
made it possible to apply various tools---%
sheaf theory, cohomology, linear algebra, for instance---%
to contextuality.
One idea that has emerged is to formulate contextuality argument in logical terms \cite{ABKLM15}:
One describes a presheaf model using logical formulas, and proves its contextuality by deriving contradiction from the formulas.

This method, however, only shows the global inconsistency of a given set of formulas;
we know them to be locally consistent only because they describe a locally consistent model that is given.
Nonetheless, when designing ways of exploiting contextuality, we may well first obtain a set of formulas or a specification, and then check if there is a model satisfying it.
This requires a logic in which consistency means local consistency.
The chief goal of this paper is to deliver such a new logic of local inference.
The two logics---%
one for global inconsistency and the other for local consistency---%
together lay a foundation for high-level logical methods of not only showing but also using contextuality as resource.

\autoref{sec:model} reviews the sheaf approach to contextuality, which takes presheaves valued in $\Sets$.
Then \autoref{sec:logic} defines what we call ``inchworm logic'', a novel logic of local inference for contextual models.
We formulate this on the basis of regular logic, since its vocabulary captures the essence of local inference.
Semantics is provided for this logic in \autoref{sec:general}, where we generalize $\Sets$-valued presheaf models to ones valued in any regular category $\DD$.
This encompasses cases that prove useful and powerful in applications:
E.g., presheaves of abelian groups, $R$-modules, etc.\ serve the purpose of cohomology;
indeed, \v{C}ech cohomology is used to detect the contextuality of $\Sets$-valued presheaves \cite{ABKLM15}.
This paper gives a uniform way of using $\DD$-valued presheaves directly as models of contextual logic.

\section{Contextual Models}\label{sec:model}

We first review the idea behind the formalism of \cite{ABKLM15}, stressing that it applies to more settings than just quantum ones.
The idea captures contextuality as a matter of topological nature, which we illustrate with a simplicial formulation equivalent to the presheaf formulation of \cite{ABKLM15}.
We also present a modification of the latter that can readily be generalized in \autoref{sec:general}.

\subsection{Topological Models for Contextuality}\label{sec:model.presheaf}

The formalism of \cite{AB11} concerns variables and values in general.
Their bare-bones structure consists of a set $X$ of variables and, for each $x \in X$, a set $A_x$ of possible values of $x$.
So we have an $X$-indexed family of sets $A_x$.
This can model various settings in which we make queries against a system and it answers, as observed in \cite{abr13,ABKLM15};
e.g.,
\begin{itemize}
\item
We measure properties $x \in X$ of a physical system and it gives back outcomes $a \in A_x$.
\item
A relational database has attributes $x \in X$, and $a \in A_x$ are possible data values for $x$.
\item
$x \in X$ are sentences of propositional logic and a set of models assign to them boolean values $a \in A_x = \bool$.
Or $x \in X$ may simply be boolean variables.
\end{itemize}
We often make a query regarding several variables in combination;
a set $U \subseteq X$ of variables the query concerns forms a \emph{context} in which the system gives back a result.
Contexts play essential r\^oles in the following two kinds of constraints, (\hyperref[item:topology.total]{a}) on answers and (\hyperref[item:topology.base]{b}) on queries.

(a)
\phantomsection%
\label{item:topology.total}%
When we make a query in a context $U$, the system returns (one or a set of) tuples $s \in \prod_{x \in U} A_x$ of values.
It then has the subset $A_U \subseteq \prod_{x \in U} A_x$ of ``admissible'' tuples that can be part of query results, and it is often the information on $A_U$ that we want.
E.g.,
\begin{itemize}
\item
From a relational database we retrieve data with an attribute list $U$, and the database returns the relation $A_U$ on sets $A_x$ ($x \in U$) as a table.
\item
Given a set of models and a set $U$ of sentences, $A_U$ is the set of combinations of values that $U$ can take;
e.g., a pair $\varphi, \lnot \varphi \in U$ only take values $(1, 0)$ or $(0, 1)$.
\item
We may measure a physical system in various states and find that some set $U$ of quantities always satisfies a certain equation that characterizes $A_U$.
\end{itemize}
(A tuple $s \in \prod_{x \in U} A_x$ is a dependent function, so it is formally a set of the form $\{\, (x, s(x)) \mid x \in U$ and $s(x) \in A_x \,\}$;
but we may refer to it as ``$(s(x), s(y), \ldots\,)$ over $(x, y, \ldots\,)$''.)

(b)
\phantomsection%
\label{item:topology.base}%
We have the family $\C \subseteq \pw X$ of contexts in which queries can be made and answered.
We may not be able to make a query in a context $V \subseteq X$ (i.e.\ $V \notin \C$) for reasons such as:
\begin{itemize}
\item
$V$ may have too many variables to deal with feasibly.
\item
A database schema may have no table encompassing all the attributes in $V$.
\item
Quantum mechanics may deem it impossible to measure all the properties in $V$ at once.
\end{itemize}

In these examples, if queries can be made in a context $U$, they can be in any $V \subseteq U$;
we also assume that queries can be made in $\{ x \}$ for any $x \in X$, but only in finite $U$.
So $\C$ is an (abstract) \emph{simplicial complex} on $X$, i.e.\ a $\subseteq$-downward closed subfamily of $\pwfin X$ with $\bigcup_{U \in \C} U = X$.
Also, if a tuple $s$ of values is admissible, so is any $t \subseteq s$.
Hence, whenever $V \subseteq U \in \C$, the projection of tuples
\begin{gather*}
\rest{-}{V} : \prod_{x \in U} A_x \to \prod_{x \in V} A_x :: s \mapsto \rest{s}{V}
\end{gather*}
restricts to $A_{V \subseteq U} : A_U \to A_V$.
Thus $A : \C^\op \to \Sets$ forms a presheaf on the poset $\C$.

In fact, $A$ is a \emph{separated} presheaf.
Generally, for any subfamily $\C$ of $\pw X$ closed under binary intersection, whenever $\bigcup_i U_i = U \in \C$ for $U_i \in \C$, a presheaf $P$ on $\C$ has the map
\begin{gather*}
\langle P_{U_i \subseteq U} \rangle_i : P_U \to \prod_i P_{U_i} :: s \mapsto (\rest{s}{U_i})_i
\end{gather*}
land in the set of \emph{matching families} for $(U_i)_i$,
\begin{gather*}
\Match{(U_i)_i}{P} = \{\, (t_i)_i \in \prod_i P_{U_i} \mid \rest{t_j}{U_j \cap U_k} = \rest{t_k}{U_j \cap U_k} \text{ for every pair } j, k \,\} .
\end{gather*}
Then $P$ is called separated if each of these $\langle P_{U_i \subseteq U} \rangle_i$ is injective, and a \emph{sheaf} if each of those $\langle P_{U_i \subseteq U} \rangle_i : P_U \to \Match{(U_i)_i}{P}$ is bijective (see \cite{MM92}).
Yet, on a simplicial complex $\C$, separated presheaves and sheaves have simpler descriptions:

\begin{fact}\label{thm:separated.presheaf}
A presheaf $P$ on a simplicial complex $\C$ is a sheaf iff $P_U = \prod_{x \in U} P_x$ for all $U \in \C$.
And $P$ is separated iff it is a subpresheaf of a sheaf, i.e.\ iff $P_U \subseteq \prod_{x \in U} P_x$ for all $U \in \C$.
\end{fact}

This shows that our $A$ above is a separated presheaf, but not generally a sheaf.
So let us write $\sPsh(\C)$ for the full subcategory of $\Sets^{\C^\op}$ of separated presheaves.
Note that every sheaf $F$ has $F_\varnothing = 1$, a singleton.
A separated presheaf $P$ has $P_\varnothing = 1$, too, unless it is the empty presheaf $U \mapsto \varnothing$, i.e.\ the model is inconsistent in every context (hence modelling, e.g., a physical system that never produces outcomes in any context of measurements).

\subsection{Presheaves and Bundles}\label{sec:model.equivalence}

The constraints (\hyperref[item:topology.total]{a}) and (\hyperref[item:topology.base]{b}) above are, indeed, matters of topology;
this idea will be useful in \autoref{sec:model.contextuality}.
Given a separated presheaf $A$ as in \autoref{sec:model.presheaf}, its underlying family of $X$-indexed sets $(A_x)_{x \in X}$ is equivalent to a set over $X$, viz.\ $\pi : \sum_{x \in X} A_x \to X :: (x, a) \mapsto x$, by $\Sets^X \simeq \Sets / X$.
The base $X$ comes with a simplicial complex $\C$, but so does $\sum_{x \in X} A_x$, taking tuples $s \in A_U$ as simplices, i.e.\ $\A = \bigcup_{U \in \C} A_U$.
And $\pi$ is a \emph{simplicial map}, or a \emph{bundle} of simplicial complexes, since $s \in A_U \subseteq \A$ implies $\pi[s] = U \in \C$.
On the other hand, any given bundle $\pi : \A \to \C$ has a family of $A_U = \{\, s \in \A \mid \pi[s] = U \,\}$ and $A_{V \subseteq U} : s \mapsto \rest{s}{V}$.

A simplicial map $\pi : \A \to \C$ is called \emph{non-degenerate} if $\rest{\pi}{s}$ is injective for every $s \in \A$.
Our $\pi$ above is non-degenerate, because every $s \in A_U$ is a \emph{local section} of the bundle $\pi$, meaning $s : U \to \sum_{x \in X} A_x$ such that $\pi \cmp s = 1_U$.
Let us write $\Simp$ and $\ndSimp$ for the categories of simplicial maps and of non-degenerate ones, respectively.
It is easy to check
that
%
%
for every simplicial complex $\C$, the slice category $\ndSimp / \C$ is a full subcategory of $\Simp / \C$;
i.e., it is the category of non-degenerate bundles and simplicial maps over $\C$.
%
%
Then, extending $\Sets^X \simeq \Sets / X$, the correspondence described above gives

\begin{fact}\label{thm:bundle.persheaf.basic}
$\sPsh(\C) \simeq \ndSimp / C$ for any simplicial complex $\C$.
\end{fact}

So here is a topological reading of (\hyperref[item:topology.total]{a}) and (\hyperref[item:topology.base]{b}).
Each $U \in \C$ is a local, small enough region of the space $X$ of variables.
The topology on the space $\A$ of values then distinguishes those tuples $s \in \prod_{x \in U} A_x$ in $A_U$ from the others and deems the former to be continuous sections.
We refer to objects of $\sPsh(\C)$ and $\ndSimp / C$ interchangeably as \emph{topological models}.

\subsection{Contextuality in Physics, Databases, and More}\label{sec:model.contextuality}

Since queries can only be made locally, i.e.\ in contexts $U \in \C$, answers to queries can only be observed locally in those contexts.
One might think this is just a matter of convenience or efficiency, since we can simply make multiple queries in multiple contexts $U_i$ to cover $X = \bigcup_i U_i$ and to recover the global information.
In certain topologies $\C$, however, this may be prevented due to contextuality.

Given a non-degenerate bundle $\pi : \A \to \C$ over a simplicial complex $\C$ on $X$, consider a \emph{global section} of it, i.e.\ $g \in \prod_{x \in X} A_x$ such that $\rest{g}{U} \in A_U$ for all $U \in \C$;
it corresponds to a matching family $g \in \Match{\C}{P}$, since $U \mapsto \prod_{x \in U} A_x$ ($U \subseteq X$) is a sheaf.
It is an assignment of values to all the variables that satisfies every constraint on combinations of values.
E.g.,
in classical logic, the models are exactly the global sections;
so the consistency of a sentence $x$ means that $(x, \true)$ is part of a global section.
Then, in the physical setting, it may seem natural to similarly think of global sections $g$ as states of the system, assigning values to all the quantities---%
so, although we can only make a query locally in a context $U \in \C$, the system in a state $g$ actually has a value $g(x)$ assigned to every quantity $x$, and the answer we receive in the context $U$ is simply $\rest{g}{U}$.
This assumption, that any section we observe is part of a context-independent global section, holds not just in classical logic but also in classical physics---%
but breaks down in quantum physics, precisely when contextuality arises.

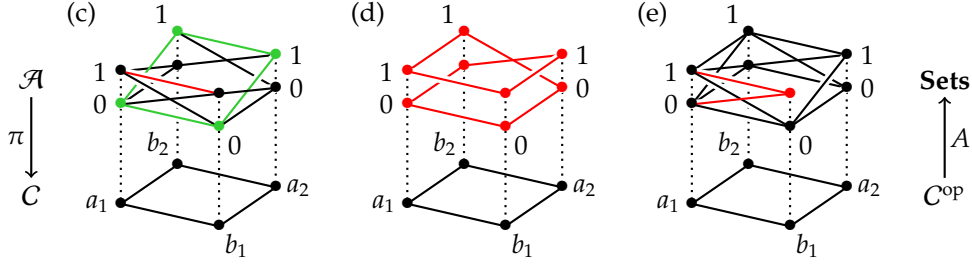
\begin{figure}
\centering
\begin{tikzpicture}[x=25pt,y=25pt,thick,label distance=-0.25em,baseline=(O.base)]
\coordinate [overlay] (O) at (0,0);
\coordinate [overlay] (T) at (0,1.5);
\coordinate [overlay] (u) at (0,0.5);
\coordinate (v0) at ($ ({-cos(1*pi/12 r)*1.2},{-sin(1*pi/12 r)*0.48}) $);
\coordinate (v1) at ($ ({-cos(7*pi/12 r)*1.2},{-sin(7*pi/12 r)*0.48}) $);
\coordinate (v2) at ($ ({-cos(13*pi/12 r)*1.2},{-sin(13*pi/12 r)*0.48}) $);
\coordinate (v3) at ($ ({-cos(19*pi/12 r)*1.2},{-sin(19*pi/12 r)*0.48}) $);
\coordinate (v0-1) at ($ (v0) + (T) $);
\coordinate (v0-0) at ($ (v0-1) + (u) $);
\coordinate (v1-1) at ($ (v1) + (T) $);
\coordinate (v1-0) at ($ (v1-1) + (u) $);
\coordinate (v2-1) at ($ (v2) + (T) $);
\coordinate (v2-0) at ($ (v2-1) + (u) $);
\coordinate (v3-1) at ($ (v3) + (T) $);
\coordinate (v3-0) at ($ (v3-1) + (u) $);
\draw (v0) -- (v1) -- (v2) -- (v3) -- (v0);
\draw [dotted] (v0-0) -- (v0);
\draw [dotted] (v1-0) -- (v1);
\draw [dotted] (v2-0) -- (v2);
\draw [dotted] (v3-0) -- (v3);
\node [inner sep=0.1em] (v0') at (v0) {$\bullet$};
\node [anchor=east,inner sep=0em] at (v0'.west) {$a_1$};
\node [inner sep=0.1em,label={[label distance=-0.625em]330:{$b_1$}}] at (v1) {$\bullet$};
\node [inner sep=0.1em] (v2') at (v2) {$\bullet$};
\node [anchor=west,inner sep=0em] at (v2'.east) {$a_2$};
\node [inner sep=0.1em,label={[label distance=-0.5em]175:{$b_2$}}] at (v3) {$\bullet$};
\draw [line width=3.2pt,white] (v0-0) -- (v3-1);
\draw [line width=3.2pt,white] (v0-1) -- (v3-0);
\draw [line width=3.2pt,white] (v0-1) -- (v3-1);
\draw (v0-0) -- (v3-1);
\draw [darkgreen] (v0-1) -- (v3-0);
\draw (v0-1) -- (v3-1);
\draw [line width=3.2pt,white] (v2-0) -- (v3-0);
\draw [line width=3.2pt,white] (v2-0) -- (v3-1);
\draw [line width=3.2pt,white] (v2-1) -- (v3-0);
\draw [darkgreen] (v2-0) -- (v3-0);
\draw (v2-0) -- (v3-1);
\draw (v2-1) -- (v3-0);
\draw [line width=3.2pt,white] (v0-0) -- (v1-0);
\draw [line width=3.2pt,white] (v0-0) -- (v1-1);
\draw [line width=3.2pt,white] (v0-1) -- (v1-0);
\draw [line width=3.2pt,white] (v0-1) -- (v1-1);
\draw [red] (v0-0) -- (v1-0);
\draw (v0-0) -- (v1-1);
\draw (v0-1) -- (v1-0);
\draw [darkgreen] (v0-1) -- (v1-1);
\draw [line width=3.2pt,white] (v2-0) -- (v1-1);
\draw [line width=3.2pt,white] (v2-1) -- (v1-0);
\draw [line width=3.2pt,white] (v2-1) -- (v1-1);
\draw [darkgreen] (v2-0) -- (v1-1);
\draw (v2-1) -- (v1-0);
\draw (v2-1) -- (v1-1);
\node [inner sep=0.1em,label=left:{$1$}] at (v0-0) {$\bullet$};
\node [inner sep=0.1em,label=left:{$0$},darkgreen] at (v0-1) {$\bullet$};
\node [inner sep=0.1em] at (v1-0) {$\bullet$};
\node [inner sep=0.1em,label={[label distance=-0.5em]330:{$0$}},darkgreen] at (v1-1) {$\bullet$};
\node [inner sep=0.1em,label=right:{$1$},darkgreen] at (v2-0) {$\bullet$};
\node [inner sep=0.1em,label=right:{$0$}] at (v2-1) {$\bullet$};
\node [inner sep=0.1em,label={[label distance=-0.5em]150:{$1$}},darkgreen] at (v3-0) {$\bullet$};
\node [inner sep=0.1em] at (v3-1) {$\bullet$};
\node [inner sep=0em] at (-1.75,2.75) {\strut(c)};
\node (C) [inner sep=0.25em] at (-2.5,0) {$\C$};
\node (A) [inner sep=0.25em] at (-2.5,1.75) {$\A$};
\draw [->] (A) -- (C) node [pos=0.5,inner sep=2pt,left] {$\pi$};
\end{tikzpicture}
\hspace{12pt}%
\begin{tikzpicture}[x=25pt,y=25pt,thick,label distance=-0.25em,baseline=(O.base)]
\coordinate [overlay] (O) at (0,0);
\coordinate [overlay] (T) at (0,1.5);
\coordinate [overlay] (u) at (0,0.5);
\coordinate (v0) at ($ ({-cos(1*pi/12 r)*1.2},{-sin(1*pi/12 r)*0.48}) $);
\coordinate (v1) at ($ ({-cos(7*pi/12 r)*1.2},{-sin(7*pi/12 r)*0.48}) $);
\coordinate (v2) at ($ ({-cos(13*pi/12 r)*1.2},{-sin(13*pi/12 r)*0.48}) $);
\coordinate (v3) at ($ ({-cos(19*pi/12 r)*1.2},{-sin(19*pi/12 r)*0.48}) $);
\coordinate (v0-1) at ($ (v0) + (T) $);
\coordinate (v0-0) at ($ (v0-1) + (u) $);
\coordinate (v1-1) at ($ (v1) + (T) $);
\coordinate (v1-0) at ($ (v1-1) + (u) $);
\coordinate (v2-1) at ($ (v2) + (T) $);
\coordinate (v2-0) at ($ (v2-1) + (u) $);
\coordinate (v3-1) at ($ (v3) + (T) $);
\coordinate (v3-0) at ($ (v3-1) + (u) $);
\draw (v0) -- (v1) -- (v2) -- (v3) -- (v0);
\draw [dotted] (v0-0) -- (v0);
\draw [dotted] (v1-0) -- (v1);
\draw [dotted] (v2-0) -- (v2);
\draw [dotted] (v3-0) -- (v3);
\node [inner sep=0.1em] (v0') at (v0) {$\bullet$};
\node [anchor=east,inner sep=0em] at (v0'.west) {$a_1$};
\node [inner sep=0.1em,label={[label distance=-0.625em]330:{$b_1$}}] at (v1) {$\bullet$};
\node [inner sep=0.1em] (v2') at (v2) {$\bullet$};
\node [anchor=west,inner sep=0em] at (v2'.east) {$a_2$};
\node [inner sep=0.1em,label={[label distance=-0.5em]175:{$b_2$}}] at (v3) {$\bullet$};
\draw [line width=3.2pt,white] (v0-0) -- (v3-0);
\draw [line width=3.2pt,white] (v0-1) -- (v3-1);
\draw [red] (v0-0) -- (v3-0);
\draw [red] (v0-1) -- (v3-1);
\draw [line width=3.2pt,white] (v2-0) -- (v3-1);
\draw [line width=3.2pt,white] (v2-1) -- (v3-0);
\draw [red] (v2-0) -- (v3-1);
\draw [red] (v2-1) -- (v3-0);
\draw [line width=3.2pt,white] (v0-0) -- (v1-0);
\draw [line width=3.2pt,white] (v0-1) -- (v1-1);
\draw [red] (v0-0) -- (v1-0);
\draw [red] (v0-1) -- (v1-1);
\draw [line width=3.2pt,white] (v2-0) -- (v1-0);
\draw [line width=3.2pt,white] (v2-1) -- (v1-1);
\draw [red] (v2-0) -- (v1-0);
\draw [red] (v2-1) -- (v1-1);
\node [inner sep=0.1em,label=left:{$1$},red] at (v0-0) {$\bullet$};
\node [inner sep=0.1em,label=left:{$0$},red] at (v0-1) {$\bullet$};
\node [inner sep=0.1em,red] at (v1-0) {$\bullet$};
\node [inner sep=0.1em,label={[label distance=-0.5em]330:{$0$}},red] at (v1-1) {$\bullet$};
\node [inner sep=0.1em,label=right:{$1$},red] at (v2-0) {$\bullet$};
\node [inner sep=0.1em,label=right:{$0$},red] at (v2-1) {$\bullet$};
\node [inner sep=0.1em,label={[label distance=-0.5em]150:{$1$}},red] at (v3-0) {$\bullet$};
\node [inner sep=0.1em,red] at (v3-1) {$\bullet$};
\node [inner sep=0em] at (-1.75,2.75) {\strut(d)};
\end{tikzpicture}
\hspace{12pt}%
\begin{tikzpicture}[x=25pt,y=25pt,thick,label distance=-0.25em,baseline=(O.base)]
\coordinate [overlay] (O) at (0,0);
\coordinate [overlay] (T) at (0,1.5);
\coordinate [overlay] (u) at (0,0.5);
\coordinate (v0) at ($ ({-cos(1*pi/12 r)*1.2},{-sin(1*pi/12 r)*0.48}) $);
\coordinate (v1) at ($ ({-cos(7*pi/12 r)*1.2},{-sin(7*pi/12 r)*0.48}) $);
\coordinate (v2) at ($ ({-cos(13*pi/12 r)*1.2},{-sin(13*pi/12 r)*0.48}) $);
\coordinate (v3) at ($ ({-cos(19*pi/12 r)*1.2},{-sin(19*pi/12 r)*0.48}) $);
\coordinate (v0-1) at ($ (v0) + (T) $);
\coordinate (v0-0) at ($ (v0-1) + (u) $);
\coordinate (v1-1) at ($ (v1) + (T) $);
\coordinate (v1-0) at ($ (v1-1) + (u) $);
\coordinate (v2-1) at ($ (v2) + (T) $);
\coordinate (v2-0) at ($ (v2-1) + (u) $);
\coordinate (v3-1) at ($ (v3) + (T) $);
\coordinate (v3-0) at ($ (v3-1) + (u) $);
\draw (v0) -- (v1) -- (v2) -- (v3) -- (v0);
\draw [dotted] (v0-0) -- (v0);
\draw [dotted] (v1-0) -- (v1);
\draw [dotted] (v2-0) -- (v2);
\draw [dotted] (v3-0) -- (v3);
\node [inner sep=0.1em] (v0') at (v0) {$\bullet$};
\node [anchor=east,inner sep=0em] at (v0'.west) {$a_1$};
\node [inner sep=0.1em,label={[label distance=-0.625em]330:{$b_1$}}] at (v1) {$\bullet$};
\node [inner sep=0.1em] (v2') at (v2) {$\bullet$};
\node [anchor=west,inner sep=0em] at (v2'.east) {$a_2$};
\node [inner sep=0.1em,label={[label distance=-0.5em]175:{$b_2$}}] at (v3) {$\bullet$};
\draw [line width=3.2pt,white] (v0-0) -- (v3-0);
\draw [line width=3.2pt,white] (v0-0) -- (v3-1);
\draw [line width=3.2pt,white] (v0-1) -- (v3-0);
\draw [line width=3.2pt,white] (v0-1) -- (v3-1);
\draw (v0-0) -- (v3-0);
\draw (v0-0) -- (v3-1);
\draw (v0-1) -- (v3-0);
\draw (v0-1) -- (v3-1);
\draw [line width=3.2pt,white] (v2-0) -- (v3-0);
\draw [line width=3.2pt,white] (v2-0) -- (v3-1);
\draw [line width=3.2pt,white] (v2-1) -- (v3-0);
\draw [line width=3.2pt,white] (v2-1) -- (v3-1);
\draw (v2-0) -- (v3-0);
\draw (v2-0) -- (v3-1);
\draw (v2-1) -- (v3-0);
\draw (v2-1) -- (v3-1);
\draw [line width=3.2pt,white] (v0-0) -- (v1-0);
\draw [line width=3.2pt,white] (v0-0) -- (v1-1);
\draw [line width=3.2pt,white] (v0-1) -- (v1-0);
\draw [line width=3.2pt,white] (v0-1) -- (v1-1);
\draw [red] (v0-0) -- (v1-0);
\draw (v0-0) -- (v1-1);
\draw [red] (v0-1) -- (v1-0);
\draw (v0-1) -- (v1-1);
\draw [line width=3.2pt,white] (v2-0) -- (v1-1);
\draw [line width=3.2pt,white] (v2-1) -- (v1-1);
\draw (v2-0) -- (v1-1);
\draw (v2-1) -- (v1-1);
\node [inner sep=0.1em,label=left:{$1$}] at (v0-0) {$\bullet$};
\node [inner sep=0.1em,label=left:{$0$}] at (v0-1) {$\bullet$};
\node [inner sep=0.1em,red] at (v1-0) {$\bullet$};
\node [inner sep=0.1em,label={[label distance=-0.5em]330:{$0$}}] at (v1-1) {$\bullet$};
\node [inner sep=0.1em,label=right:{$1$}] at (v2-0) {$\bullet$};
\node [inner sep=0.1em,label=right:{$0$}] at (v2-1) {$\bullet$};
\node [inner sep=0.1em,label={[label distance=-0.5em]150:{$1$}}] at (v3-0) {$\bullet$};
\node [inner sep=0.1em] at (v3-1) {$\bullet$};
\node [inner sep=0em] at (-1.75,2.75) {\strut(e)};
\node (C) [inner sep=0.25em] at (2.625,0) {$\C^\op$};
\node (A) [inner sep=0.25em] at (2.625,1.75) {$\Sets$};
\draw [->] (C) -- (A) node [pos=0.5,inner sep=2pt,right] {$A$};
\end{tikzpicture}
\caption{Bundles for (c) the Hardy model and (d) the PR-box}
\label{fig:Hardy.PR}%
\label{fig:no-signalling}%
\label{pic:Hardy}%
\label{pic:PR}%
\label{pic:signal}%
\end{figure}
\autoref{fig:Hardy.PR} shows ``Bell-type'' scenarios in which Alice and Bob measure properties of a system, perhaps a quantum one.
The base $\C$ expresses constraints of type (\hyperref[item:topology.base]{b}) above:
Alice can make at most one of two measurements $a_1$ and $a_2$ at a time, so she chooses one;
similarly Bob chooses from $b_1$ and $b_2$---%
so there are four possible combinations of measurements, indicated by the four edges of $\C$.
Alice and Bob repeat measurements in different contexts, and learn that each $x \in X = \{ a_1, a_2, b_1, b_2 \}$ has two possible outcomes $0$ and $1$, but that some combinations of outcomes are never obtained.
$\A$ expresses these constraints, of type (\hyperref[item:topology.total]{a}), with edges indicating possible combinations.
E.g., $\A$ of (\hyperref[pic:Hardy]{c}) deems every joint outcome of $(a_1, b_1)$ possible, with $A_{\{ a_1, b_1 \}} = \2 \times \2$;
but $(0, 0)$ is not a possible joint outcome of $(a_2, b_2)$.

The models in \autoref{fig:Hardy.PR} all violate the classical assumption above, and are examples of

\begin{definition}\label{def:contextuality.sets}
A topological model is said to be \emph{logically contextual} if not all of its local sections extend to global ones, and \emph{strongly contextual} if it has no global section at all (see \cite{ABKLM15}).
\end{definition}

(\hyperref[pic:Hardy]{c}) of \autoref{fig:Hardy.PR} represents an example of logical contextuality due to \cite{har93} that is realizable in quantum physics.
It has several global sections, e.g.\ the one marked in green;
call it $g$.
So, when Alice and Bob measure $(a_1, b_1)$ and observe $(0, 0)$, the classical explanation is possible
that the system was in the state $g$ and had outcomes $g(x)$ assigned to all the measurements $x \in X$,
and
that Alice and Bob have simply retrieved that information on $U$.
On the other hand, the local section in red, $(1, 1)$ over $(a_1, b_1)$, does not extend to any global section.
This means that the classical explanation is simply impossible for this joint outcome.
Furthermore, the classical explanation is never possible in the strongly contextual (\hyperref[pic:PR]{d}).
This model, called the \emph{PR box} \cite{PR94}, is not quantum-realizable (though it plays an important r\^ole in the quantum information literature), but quantum physics exhibits many other instances of strong contextuality.

The upshot is that contextuality consists in \emph{global inconsistency coupled with local consistency}:
A section $s \in A_U$ is consistent locally, in the sense of satisfying the constraint on query results in the context $U$, but it may be inconsistent globally, in the sense of contradicting all the other constraints and thereby failing to extend to a global section.

The general definition of contextuality in terms of global sections can also be applied to relational databases:
Contextuality then corresponds exactly to the absence of a universal relation \cite{abr13}.
In fact, the natural join
\begin{gather*}
\natjoin_{U \in \C} A_U = \{\, g \in \prod_{x \in X} A_x \mid \rest{g}{U} \in A_U \text{ for all } U \in \C \,\}
\end{gather*}
of relations $A_U$ (which is the largest of universal relations if there are any) is, simply by definition, the set of global sections.
Hence, e.g.\ in \autoref{fig:Hardy.PR}, the sections in red are lost from $\natjoin_{U \in \C} A_U$.

\subsection{No-Signalling Principle}\label{sec:model.no-signalling}

Local consistency means, partly, that a local section may exist without extending to global.
Yet it involves more---%
viz.\ a constraint that is called the \emph{no-signalling} principle in the physical setting \cite{GRW80}.
For a topological model $A$, it amounts to the condition that every $A_{U \subseteq V} : A_V \to A_U :: s \mapsto \rest{s}{U}$ (i.e.\ the projection of admitted tuples, or the restriction of sections) is a surjection.

An example violating no-signalling is (\hyperref[pic:signal]{e}) of \autoref{fig:no-signalling}:
$A_{\{ b_1 \} \subseteq \{ a_2, b_1 \}} : A_{\{ a_2, b_1 \}} \to A_{\{ b_1 \}}$ is not surjective.
Suppose Alice and Bob make measurements, Bob chooses to measure $b_1$, and he observes $1$, which is not in the image of $A_{\{ b_1 \} \subseteq \{ a_2, b_1 \}}$.
This means that Bob has received the signal from Alice (no matter how far away she may be!)\ that she has chosen $a_1$ and not $a_2$.

To see why no-signalling should be part of local consistency, regard $\A$ in (\hyperref[pic:signal]{e}) as representing a relational database.
It has tables $A_{\{ a_1, b_1 \}}$ and $A_{\{ a_2, b_1 \}}$;
but, when queried about the attribute $b_1$, they yield different results of projection, differing in whether $1$ is in or not.
Thus, no-signalling means the consistency of projections (see \cite{abr13}).
Indeed, as we will see in \autoref{sec:general.interpretation}, no-signalling means a sort of coherence of $A$ as a semantic model of logic.

\begin{definition}\label{def:no-signalling.sets}
We say that
a separated presheaf $A : \C^\op \to \Sets$ is no-signalling if it satisfies \eqref{item:no-signalling.presheaf},
and that
a non-degenerate bundle $\pi : \A \to \C$ is no-signalling if it satisfies \eqref{item:no-signalling.bundle}:
\begin{enumerate}
\setcounter{enumi}{\value{equation}}
\item\label{item:no-signalling.presheaf}
Every $A_{U \subseteq V} : A_V \to A_U$ is a surjection.
\item\label{item:no-signalling.bundle}
If $\pi[s] \subseteq U$ for $s \in \A$ and $U \in \C$, then there is some $t \in \A$ such that $s \subseteq t$ and $\pi[t] = U$.
\setcounter{equation}{\value{enumi}}
\end{enumerate}
\end{definition}

Clearly, \eqref{item:no-signalling.presheaf} and \eqref{item:no-signalling.bundle} coincide via $\sPsh(\C) \simeq \ndSimp / \C$ (\autoref{thm:bundle.persheaf.basic}).
Hence their full subcategories of no-signalling models are equivalent.
Note that \eqref{item:no-signalling.presheaf} or \eqref{item:no-signalling.bundle} implies $A_U \neq \varnothing$ for all $U \in \C$, if $A_\varnothing \neq \varnothing$.
So, while the empty model is no-signalling, all the other, nonempty no-signalling models (which are, essentially, the ``empirical models'' of \cite{ABKLM15}) are locally consistent.

\section{Contextual Logics}\label{sec:logic}

After reviewing a kind of contextuality argument for topological models, we explain why the logic of such argument is \emph{not} supposed to be sound with respect to those models, and why we need another logic, viz.\ the logic of topological models.
We then introduce our candidate for such a logic.

\subsection{Contextuality Argument:\ Logic of Global Inconsistency}\label{sec:logic.global}

Viewing $A_U$ as representing a constraint on assignments of values to variables $x \in U$, we can describe a topological model $A$ using formulas in contexts $U \in \C$ of variables.
E.g., the assignments of $(0, 0)$ and $(1, 1)$ to $(x, y)$ satisfy the equation $x \oplus y = 0$, where $\oplus$ is for XOR, i.e.\ addition modulo $2$;
the assignments $(0, 1)$ and $(1, 0)$ satisfy $x \oplus y = 1$.
Therefore the PR box, (\hyperref[pic:PR]{d}) of \autoref{fig:Hardy.PR}, satisfies the following set of equations:
\begin{gather}\label{eq:AvN.PR}
a_1 \oplus b_1 = 0 , \mspace{18mu}
a_1 \oplus b_2 = 0 , \mspace{18mu}
a_2 \oplus b_1 = 0 , \mspace{18mu}
a_2 \oplus b_2 = 1
\end{gather}
These are in fact inconsistent:
Their right-hand sides sum to $1$, but the left to $0$ regardless of the values of variables (since each variable occurs twice).
This is to say that no global assignment of values satisfies all the constraints of $A_U$, i.e., that $A$ is strongly contextual.

A family of arguments of this sort, using XOR (or parity) equations, has been given to show the strong contextuality of a range of quantum examples;
the first instance in literature was in \cite{mer90} for the GHZ state \cite{GHZ89}.
This sort of so-called ``all-vs-nothing argument'' was formalized and generalized in \cite{ABKLM15}.
On the other hand, one may also adopt more expressive languages, such as Boolean formulas, to express a wider range of constraints.

Formulas can also be used to show logical (and not strong) contextuality.
E.g., the Hardy model, (\hyperref[pic:Hardy]{c}) of \autoref{fig:Hardy.PR}, satisfies the antecedents of
\begin{gather}\label{eq:AvN.Hardy}
\lnot a_1 \vee \lnot b_2 , \mspace{18mu}
\lnot a_2 \vee \lnot b_1 , \mspace{18mu}
a_2 \vee b_2 \mspace{18mu}
\vdash \mspace{18mu}
\lnot a_1 \vee \lnot b_1
\end{gather}
but not the consequent, due to the contextual section $(1, 1)$ over $(a_1, b_1)$.
This shows that this local section can be part of no global assignment satisfying all the constraints.

Yet this kind of contextuality argument needs some reflection.
The inconsistency of a set $\Gamma$ of formulas, $\Gamma \vdash \bot$, does not mean that $\Gamma$ has no model;
in fact, the PR box, (\hyperref[pic:PR]{d}) of \autoref{fig:Hardy.PR}, satisfies all the equations in \eqref{eq:AvN.PR}.
In the same vein, the derivability $\Gamma \vdash \varphi$ does not mean that every model of $\Gamma$ satisfies $\varphi$;
the Hardy model (\hyperref[pic:Hardy]{c}) satisfies $\Gamma$ but not $\varphi$ of \eqref{eq:AvN.Hardy}.
So the logic of $\vdash$ here is \emph{not} sound with respect to contextual models---%
indeed, that is the whole point of the argument.
Invalidating $\vdash$ precisely means contextuality:
$\Gamma \vdash \bot$ really means that no global section satisfies $\Gamma$;
it is why any model of $\Gamma$ has no global section.
$\Gamma \vdash \varphi$ means that every global section satisfying $\Gamma$ satisfies $\varphi$;
it is why any model satisfying $\Gamma$ but not $\varphi$ must have local sections (viz.\ ones not satisfying $\varphi$) that fail to extend to global sections.

In this sense, the logic of $\vdash$ here is a ``global logic'' of global sections.
We should then note that this logic, by itself, says very little about local consistency.
To see this, consider:
\begin{gather}\label{eq:AvN.inconsistent}
a_1 \oplus b_1 = 0 , \mspace{18mu}
a_1 \oplus b_1 = 1
\end{gather}
This set of equations is, like \eqref{eq:AvN.PR}, inconsistent.
It is, however, inconsistent not just globally but also locally:
Not only does no global section satisfy both equations, no local section over the context $\{ a_1, b_1 \}$ does;
a model $A$ satisfies \eqref{eq:AvN.inconsistent} only if $A_{\{ a_1, b_1 \}} = \varnothing$
(the physical system can give no outcomes to the measurements $a_1$, $b_1$;
the sentences $a_1$, $b_1$ are not just inconsistent but can have no truth values).
Yet the global logic does not tell us why \eqref{eq:AvN.PR} is locally consistent whereas \eqref{eq:AvN.inconsistent} is not.
Thus the kind of argument above is really a ``global-inconsistency argument'':
It shows contextuality only because we already know the formulas to be locally consistent, having obtained them as descriptions of some model.

\subsection{``Inchworm Logic'' of Local Inference}\label{sec:logic.inchworm}

In contrast to \eqref{eq:AvN.PR}, $\Gamma \vdash \bot$ of \eqref{eq:AvN.inconsistent} means local inconsistency over $\{ a_1, b_1 \}$ since both equations in $\Gamma$ are in the context $\{ a_1, b_1 \}$.
Turning $\Gamma, \varphi \vdash \bot$ into the form of inference, if $\Gamma, \varphi$ are in a context $U$, then $\Gamma \vdash \lnot \varphi$ gives local entailment over $U$.
E.g., the antecedents of
\begin{gather*}
a_1 = 0 , \mspace{18mu}
b_1 = 0 \mspace{18mu}
\vdash \mspace{18mu}
a_1 \oplus b_1 = 0
\end{gather*}
rule out all the sections over $(a_1, b_1)$ except $(0, 0)$, which satisfies the consequent.

Indeed, local inference can be carried out across different contexts, validly in no-signalling models, subject to one constraint.
To see this, expand the base $\C$ in \autoref{fig:Hardy.PR} from (\hyperref[pic:Alice.Bob]{f}) of \autoref{fig:Charlie.inchworm} to (\hyperref[pic:Alice.Bob.Charlie]{g}), where the four triangles are in $\C$---%
so a new experimenter, Charlie, can make his measurement $c$ along with Alice and Bob.
\begin{figure}
\centering
\begin{tikzpicture}[x=25pt,y=25pt,thick,label distance=-0.25em,baseline=(O.base)]
\coordinate [overlay] (O) at (0,0);
\coordinate (v0) at ($ ({-cos(1*pi/12 r)*1.2},{-sin(1*pi/12 r)*0.48}) $);
\coordinate (v1) at ($ ({-cos(7*pi/12 r)*1.2},{-sin(7*pi/12 r)*0.48}) $);
\coordinate (v2) at ($ ({-cos(13*pi/12 r)*1.2},{-sin(13*pi/12 r)*0.48}) $);
\coordinate (v3) at ($ ({-cos(19*pi/12 r)*1.2},{-sin(19*pi/12 r)*0.48}) $);
\draw (v0) -- (v1) -- (v2) -- (v3) -- (v0);
\node [inner sep=0.1em] (v0') at (v0) {$\bullet$};
\node [anchor=east,inner sep=0em] at (v0'.west) {$a_1$};
\node [inner sep=0.1em,label={[label distance=-0.625em]330:{$b_1$}}] at (v1) {$\bullet$};
\node [inner sep=0.1em] (v2') at (v2) {$\bullet$};
\node [anchor=west,inner sep=0em] at (v2'.east) {$a_2$};
\node [inner sep=0.1em,label={[label distance=-0.5em]175:{$b_2$}}] at (v3) {$\bullet$};
\node [inner sep=0em] at (-1.5,1.5) {\strut(f)};
\node (C) [inner sep=0.25em] at (-2.5,0) {$\C$};
\end{tikzpicture}
\hspace{12pt}%
\begin{tikzpicture}[x=25pt,y=25pt,thick,label distance=-0.25em,baseline=(O.base)]
\coordinate [overlay] (O) at (0,0);
\coordinate (v0) at ($ ({-cos(1*pi/12 r)*1.2},{-sin(1*pi/12 r)*0.48}) $);
\coordinate (v1) at ($ ({-cos(7*pi/12 r)*1.2},{-sin(7*pi/12 r)*0.48}) $);
\coordinate (v2) at ($ ({-cos(13*pi/12 r)*1.2},{-sin(13*pi/12 r)*0.48}) $);
\coordinate (v3) at ($ ({-cos(19*pi/12 r)*1.2},{-sin(19*pi/12 r)*0.48}) $);
\coordinate (v4) at (0,1.2);
\draw [opacity=0.5] (v2) -- (v3) -- (v0);
\draw [opacity=0.5] (v3) -- (v4);
\node [inner sep=0.1em,opacity=0.5] at (v3) {$\bullet$};
\node [inner sep=0.5em,below,opacity=0.5] at (v3.south) {$b_2$};
\fill [darkgreen!50,opacity=0.5] (v0) -- (v4) -- (v3) -- (v0);
\fill [darkgreen!50,opacity=0.5] (v2) -- (v4) -- (v3) -- (v2);
\fill [darkgreen!50,opacity=0.5] (v0) -- (v4) -- (v1) -- (v0);
\draw (v1) -- (v0) -- (v4);
\fill [darkgreen!50,opacity=0.5] (v2) -- (v4) -- (v1) -- (v2);
\draw (v1) -- (v2) -- (v4);
\draw (v1) -- (v4);
\node [inner sep=0.1em,label=left:{$a_1$}] at (v0) {$\bullet$};
\node [inner sep=0.1em,label={[label distance=-0.5em]330:{$b_1$}}] at (v1) {$\bullet$};
\node [inner sep=0.1em,label=right:{$a_2$}] at (v2) {$\bullet$};
\node [inner sep=0.1em,label=above:{$c$}] at (v4) {$\bullet$};
\node [inner sep=0em] at (-1.5,1.5) {\strut(g)};
\end{tikzpicture}
\hspace{12pt}%
\begin{tikzpicture}[x=25pt,y=25pt,thick,label distance=-0.25em,baseline=(O.base)]
\coordinate [overlay] (O) at (0,0);
\coordinate [overlay] (r) at (1,0);
\coordinate (v0) at ($ ({-cos(1*pi/12 r)*1.2},{-sin(1*pi/12 r)*0.48}) $);
\coordinate (v1) at ($ ({-cos(7*pi/12 r)*1.2},{-sin(7*pi/12 r)*0.48}) $);
\coordinate (v1') at ($ (v1) + (r) $);
\coordinate (v1'') at ($ (v1') + (r) $);
\coordinate (v2) at ($ ({-cos(13*pi/12 r)*1.2},{-sin(13*pi/12 r)*0.48}) + (r) + (r) $);
\coordinate (v4) at (0,1.2);
\coordinate (v4') at ($ (v4) + (r) $);
\coordinate (v4'') at ($ (v4') + (r) $);
\fill [darkgreen!50,opacity=0.5] (v0) -- (v4) -- (v1) -- (v0);
\draw (v0) -- (v4) node [pos=0.66,inner sep=0.75em,left] {$U$};
\draw (v4) -- (v1) -- (v0);
\draw (v1') -- (v4');
\fill [darkgreen!50,opacity=0.5] (v2) -- (v4'') -- (v1'') -- (v2);
\draw (v2) -- (v4'') node [pos=0.66,inner sep=0.75em,right] {$V$};
\draw (v4'') -- (v1'') -- (v2);
\node [inner sep=0.1em,label=left:{$a_1$}] at (v0) {$\bullet$};
\node [inner sep=0.1em,label={[label distance=-0.5em]330:{$b_1$}}] at (v1) {$\bullet$};
\node [inner sep=0.1em,label={[label distance=-0.5em]330:{$b_1$}}] at (v1') {$\bullet$};
\node [inner sep=0.1em,label={[label distance=-0.5em]330:{$b_1$}}] at (v1'') {$\bullet$};
\node [inner sep=0.1em,label=right:{$a_2$}] at (v2) {$\bullet$};
\node [inner sep=0.1em,label=above:{$c$}] at (v4) {$\bullet$};
\node [inner sep=0.1em,label=above:{$c$}] at (v4') {$\bullet$};
\node [inner sep=0.1em,label=above:{$c$}] at (v4'') {$\bullet$};
\coordinate (C) at ($ (v1')!0.5!(v4') $);
\coordinate (A) at ($ (v1)!0.5!(v4)!0.35!(v0) $);
\coordinate (B) at ($ (C) + (-0.5,0) $);
\coordinate (D) at ($ (C) + (0.5,0) $);
\coordinate (E) at ($ (v1'')!0.5!(v4'')!0.35!(v2) $);
\draw [->,red] (A) .. controls ($ (A)!0.5!(B) + (0,0.25) $) .. (B);
\draw [->,rounded corners=2pt,red] ($ (C) + (-0.35,0) $) -- ($ (C) + (-0.075,0) $) .. controls ($ (C) + (-0.4,0.5) $) and ($ (C) + (0.4,0.5) $) .. ($ (C) + (0.075,0) $) -- ($ (C) + (0.35,0) $);
\draw [->,red] (D) .. controls ($ (D)!0.5!(E) + (0,0.25) $) .. (E);
\node [inner sep=0em] at (-1.5,1.5) {\strut(h)};
\end{tikzpicture}
\caption{Charlie and an inchworm}
\label{fig:Charlie.inchworm}
\label{pic:Alice.Bob}
\label{pic:Alice.Bob.Charlie}
\label{pic:Alice.Bob.Charlie.inchworm}
\end{figure}
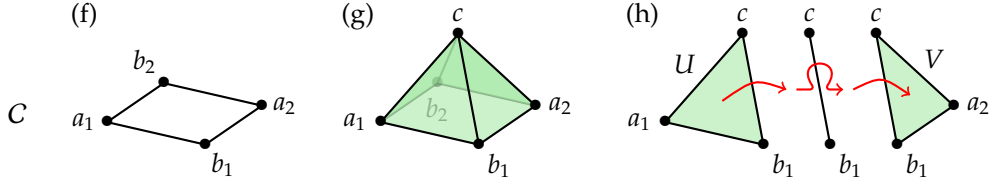
Now rewrite the locally consistent \eqref{eq:AvN.PR} in the inference form \eqref{eq:AvN.PR.inference} (we replace $x \oplus y = 0$ with simpler $x = y$) and compare it to \eqref{eq:AvN.PR.Charlie}:
\begin{align}
\label{eq:AvN.PR.inference}
&&&&&&&&&&&&&&&&
a_1 & = b_1 , &
a_1 & = b_2 , &
a_2 & = b_1 &
& \vdash &
a_2 & = b_2
&&&&&&&&&&&&&&&&
\\
\label{eq:AvN.PR.Charlie}
&&&&&&&&&&&&&&&&
a_1 & = b_1 , &
a_1 & = c , &
a_2 & = b_1 &
& \vdash &
a_2 & = c
&&&&&&&&&&&&&&&&
\end{align}
\eqref{eq:AvN.PR.Charlie} is valid in no-signalling models, whereas \eqref{eq:AvN.PR.inference} is not (the PR box is a countermodel, as it is a model for \eqref{eq:AvN.PR}).
The only difference is $c$ replacing $b_2$---%
this tiny difference, however, enables us to obtain \eqref{eq:AvN.PR.Charlie} in the following two steps:
\begin{gather}\label{eq:AvN.PR.Charlie.steps}
\AXC{$a_1 = b_1$}
\AXC{$a_1 = c$}
\RL{$U$}
\BIC{$b_1 = c$}
\AXC{$a_2 = b_1$}
\RL{$V$}
\BIC{$a_2 = c$}
\DP
\end{gather}
The first step is within the context $U = \{ a_1, b_1, c \}$, hence valid locally:
Every section over $U$ satisfying the antecedents satisfies the consequent.
Similarly, the second step is valid within $V = \{ a_2, b_1, c \}$.
The key aspect is that the formula in the middle, $b_1 = c$, can be in the context $U \cap V$ and so in $U$ or in $V$.
The upshot is that information gets passed on from a larger context $U$ to a smaller $U \cap V$ and then to another larger $V$---%
just like the locomotion of an inchworm, if (\hyperref[pic:Alice.Bob.Charlie.inchworm]{h}) of \autoref{fig:Charlie.inchworm} helps to visualize it.
Crucially, the no-signalling property is essential when the inchworm moves from a larger context to a smaller:
E.g., the first step of \eqref{eq:AvN.PR.Charlie.steps} concludes that every section over $\{ a_1, b_1, c \}$ satisfies $b_1 = c$;
but then, in the absence of no-signalling, there may be a section over $\{ b_1, c \}$ violating $b_1 = c$ without extending to $\{ a_1, b_1, c \}$.
We will discuss the semantic r\^ole of no-signalling further in \autoref{sec:general.interpretation}.

\subsection{Formalizing the Inchworm}\label{sec:logic.formalize}

We formalize and generalize the idea of ``inchworm inference''.
As in the example in \autoref{sec:logic.inchworm}, an inchworm logic is obtained by constraining a global logic.
We assume this logic to be (at least) regular, i.e.\ to have $\top$, $\wedge$, and $\exists$, for the reasons explained shortly.

\begin{definition}\label{def:inchworm.logic}
Let $\L$ be a language of regular logic (or richer) whose variables include $X$.
For each $x \in X$, write $T_x$ for the type of $x$, and then, for each $\bar{x} \subseteq X$, write $\Phi_{\bar{x}}$ for the set of formulas in the context $\bar{x} : T_{\bar{x}}$.
Given a simplicial complex $\C$ on $X$, the \emph{$\C$-contextual fragment} of $\L$ is $\Phi_\C = \bigcup_{U \in \C} \Phi_U$.
By a \emph{$\C$-contextual language} $\L_\C$, we simply mean a pair of such $\L$ and $\Phi_\C$.
Now let $\TT$ be a regular theory in $\L$ given by an entailment relation $\vdash$ (which is \emph{not} required to be binary).
Then the inchworm fragment of $\vdash$ in $\L_\C$ is the entailment relation $\vdash_\C$ on $\Phi_\C$ defined inductively by the following.
We write $\Gamma_U = \Gamma \cap \Phi_U$.
\begin{enumerate}
\setcounter{enumi}{\value{equation}}
\item\label{item:inchworm.base}
$\Gamma \vdash_\C \varphi$ if there is $U \in \C$ such that $\varphi \in \Phi_U$ and $\Gamma_U \vdash \varphi$.
\item\label{item:inchworm.inductive}
If $\Gamma \vdash_\C \varphi$ and $\Delta, \varphi \vdash_\C \psi$ then $\Gamma, \Delta \vdash_\C \psi$.
\setcounter{equation}{\value{enumi}}
\end{enumerate}
\end{definition}

\eqref{item:inchworm.base} expresses the idea that $\vdash$ within a single context is valid locally as well.
In \eqref{item:inchworm.inductive}, note that the first two instances of $\vdash_\C$ may be witnessed by different contexts.
Observe also that $\Gamma \vdash_\C \varphi$ entails $\Gamma \vdash \varphi$;
thus $\vdash_\C$ is a fragment of $\vdash$.

\begin{example}\label{eg:AvN.syntax}
Let $\L$ have $T$ as a basic type;
$0$, $1$ be constants of type $T$;
and $\oplus$ be a function symbol of type $\oplus : T \times T \to T$.
Let $T_x = T$ for all $x \in X$.
So, e.g., $x : T, y : T \mid x \oplus y = 1$ makes sense, and is in $\Phi_{\{ x, y \}}$.
This gives equations of the kind relevant to the examples in \autoref{sec:logic.global}.
Note that $\Phi_\C$ is a union.
E.g., for (\hyperref[pic:Alice.Bob]{f}) of \autoref{fig:Charlie.inchworm}, $a_i = 0$ is in $\Phi_{\{ a_i \}} \subseteq \Phi_\C$ for both $i = 1, 2$, but $a_1 = 0 \wedge a_2 = 0$ is \emph{not} in $\Phi_\C$ since $\{ a_1, a_2 \} \notin \C$.
\end{example}

We assume $\L$ to have $\top$ and $\wedge$, so that pieces of information can be combined within the same context.
The inchworm moves from a smaller context $U$ to a larger $V$ via the order embedding $i : (\Phi_U, \vdash_U) \incto (\Phi_V, \vdash_V)$, and from $V$ to $U$ via the left adjoint $\exists_{V \setminus U}$ of $i$.
\begin{gather*}
\begin{tikzpicture}[x=25pt,y=25pt,baseline=(C.base)]
\node (C) [anchor=east,inner sep=0.25em] at (0,0) {$\Phi_V$};
\node (D) [anchor=west,inner sep=0.25em] at (2,0) {$\Phi_U$};
\draw [transform canvas={yshift=5pt},->] (C) -- (D) node [pos=0.5,inner sep=2pt,above] {$\exists_{V \setminus U}$};
\draw [transform canvas={yshift=-5pt},right hook->] (D) -- (C) node [pos=0.5,inner sep=2pt,below] {$i$};
\path (C) -- (D) node [pos=0.5,rotate=-90] {$\dashv$};
\end{tikzpicture}
\end{gather*}
Then $\exists_{V \setminus U} \dashv i$ means that, for any $\varphi \in \Phi_V$, $\exists_{V \setminus U} \ldot \varphi \in \Phi_U$
encapsulates all and only the information that $\varphi$ entails on $U$.
We also have $\exists_{V \setminus U} \cmp i \cong 1$, so a piece of information that can be both about $U$ and about $V$ undergoes no change when carried across $U$ and $V$.

\section{Contextual Semantics in Regular Categories}\label{sec:general}

Our definition of model using a sheaf generalizes by replacing $\Sets$ with any category $\DD$ with finite limits, since the base $\C$ is a simplicial complex.
Yet, for the sake of no-signalling, we moreover need $\DD$ to be regular.
References on regular categories and their categorical logic include \cite{oos95,but98,joh02}.
We then lay out how to model the inchworm logic in $\DD$.

\subsection{Topological Models in Regular Categories}\label{sec:general.presheaf}

Let $\C$ be a simplicial complex on a set $X$, and $\DD$ be a category with finite limits.
By a presheaf on $\C$ valued in $\DD$, we mean any contravariant functor $P : \C^\op \to \DD$.
Then the definitions of separated presheaf and sheaf generalize straightforwardly to

\begin{definition}\label{def:separated.presheaf.general}
We say that an $\DD$-valued presheaf $P$ on $\C$ is separated if the arrow
\begin{gather*}
\langle P_{U_i \subseteq U} \rangle_i : P_U \to \prod_i P_{U_i}
\end{gather*}
is monic whenever $\bigcup_i U_i = U$, and a sheaf if, whenever $\bigcup_i U_i = U$, $\langle P_{U_i \subseteq U} \rangle_i$ is an equalizer as follows, where $p_j : \prod_i P_{U_i} \to P_{U_j}$ and $p_k : \prod_i P_{U_i} \to P_{U_k}$ are the projections.
\begin{gather*}
\begin{tikzpicture}[x=25pt,y=25pt,baseline=(A.base)]
\node (A) [anchor=east,inner sep=0.25em] at (0,0) {$P_U$};
\node (B) [anchor=west,inner sep=0.25em] at ($ (A.east) + (3,0) $) {$\prod_i P_{U_i}$};
\node (C) [anchor=west,inner sep=0.25em] at ($ (B.east) + (5,0) $) {$\prod_{j, k} P_{U_j \cap U_k}$};
\draw [>->] (A) -- (B) node [pos=0.5,inner sep=2pt,above] {$\langle P_{U_i \subseteq U} \rangle_i$};
\draw [transform canvas={yshift=3.2pt},->] (B) -- (C) node [pos=0.5,inner sep=2pt,above] {$\langle P_{U_j \cap U_k \subseteq U_j} \cmp p_j \rangle_{j, k}$};
\draw [transform canvas={yshift=-3.2pt},->] (B) -- (C) node [pos=0.5,inner sep=2pt,below] {$\langle P_{U_j \cap U_k \subseteq U_k} \cmp p_k \rangle_{j, k}$};
\end{tikzpicture}
\end{gather*}
\end{definition}

Again, every sheaf $F$ has $F_\varnothing = 1$, the terminal object of $\DD$, and every separated presheaf $P$ has $P_\varnothing \monoto 1$.
Also, for simpler descriptions, \autoref{thm:separated.presheaf} generalizes to

\begin{fact}\label{thm:separated.presheaf.general}
An $\DD$-valued presheaf $P$ on a simplicial complex $\C$ is a sheaf iff $P_U = \prod_{x \in U} P_x$ for all $U \in \C$.
And $P$ is separated iff it is a subpresheaf of a sheaf, i.e.\ iff each $\langle P_{\{ x \} \subseteq U} \rangle_{x \in U} : P_U \to \prod_{x \in U} P_x$ is monic, i.e.\ iff each $P_U$ is a relation in $\DD$ on $(P_x)_{x \in U}$.
\end{fact}

Next we define the no-signalling property for $\DD$-valued separated presheaves.
In doing so, we need to choose from several generalizations of the notion of surjection in \eqref{item:no-signalling.presheaf} of \autoref{def:no-signalling.sets};
the one that serves our purpose is the one that provides semantics for $\exists : \Phi_V \rightleftarrows \Phi_U : i$ in the inchworm logic.
This is the principal reason we need $\DD$ to be regular;
then, in $\DD$, every arrow $f : C \to D$ gives rise to the adjoint pair
\begin{gather*}
\begin{tikzpicture}[x=25pt,y=25pt,baseline=(C.base)]
\node (C) [anchor=east,inner sep=0.25em] at (0,0) {$\Sub_\DD(C)$};
\node (D) [anchor=west,inner sep=0.25em] at (2,0) {$\Sub_\DD(D)$};
\draw [transform canvas={yshift=5pt},->] (C) -- (D) node [pos=0.5,inner sep=2pt,above] {$\exists_f$};
\draw [transform canvas={yshift=-5pt},->] (D) -- (C) node [pos=0.5,inner sep=2pt,below] {$f^{-1}$};
\path (C) -- (D) node [pos=0.5,rotate=-90] {$\dashv$};
\end{tikzpicture}
\end{gather*}
(e.g.\ \cite[Lemma 2.5]{but98}), and moreover $\exists_f \cmp f^{-1} = 1_{\Sub(D)}$ (and so $f^{-1}$ is an order embedding) if $f$ is a regular epi (essentially, \cite[Corollary D1.2.8]{joh02}).
Therefore the right generalization of \autoref{def:no-signalling.sets} is the following \autoref{def:no-signalling.general}, with an alternative description in \autoref{thm:no-signalling.subpresheaf}.

\begin{definition}\label{def:no-signalling.general}
A separated presheaf $A$ on a simplicial complex $\C$ valued in a regular category $\DD$ is said to be no-signalling if every $A_{U \subseteq V} : A_V \to A_U$ is a regular epi.
\end{definition}

\begin{fact}\label{thm:no-signalling.subpresheaf}
Let $F$ be a sheaf on a simplicial complex $\C$.
Then a family $(i_U : A_U \monoto F_U)_{U \in \C}$ of subobjects forms a subpresheaf of $F$, and hence a separated presheaf, iff $A_V \leqslant {F_{U \subseteq V}}^{-1}(A_U)$, or equivalently $\exists_{F_{U \subseteq V}}(A_V) \leqslant A_U$, whenever $U \subseteq V \in \C$.
Moreover, a separated presheaf $i : A \monoto F$ is no-signalling iff $\exists_{F_{U \subseteq V}}(A_V) = A_U$ whenever $U \subseteq V \in \C$.
\end{fact}

\subsection{Global Inconsistency in Regular Categories}\label{sec:general.inconsistency}

\autoref{def:contextuality.sets} of contextuality for $\Sets$-valued presheaves can now extend to ones valued in any regular category $\DD$.
Let $A$ be an $\DD$-valued separated presheaf on a simplicial complex $\C$ on $X$.
It is a subpresheaf of a sheaf $F$ on $\C$.
In fact, let us assume, just in this subsection, that $X$ is finite (or that $\DD$ is complete);
then, by \autoref{thm:separated.presheaf.general} (or a straightforward generalization), $F$ extends uniquely to a sheaf on $\pw X$, viz.\ $F : U \mapsto \prod_{x \in U} F_x$.
Then the set of global sections of $A$---%
i.e.\ the natural join
\begin{gather*}
\natjoin_{U \in \C} A_U = \{\, g \in \prod_{x \in X} A_x \mid \rest{g}{U} \in A_U \text{ for all } U \in \C \,\}
\end{gather*}
of the relations $A_U \subseteq F_U$---%
generalizes to the $\DD$-valued case:

\begin{fact}
Given any $\DD$-valued separated presheaf $A$, let $F$ be a sheaf such that $i : A \monoto F$ and, using $A_U$ as predicates in the internal language of $\DD$, define
\begin{gather*}
\natjoin A = \Scott{\, \bar{x} : F_X \mid \bigwedge_{U \in \C} A_U(F_{U \subseteq X} \bar{x}) \,} = \bigwedge_{U \in \C} {F_{U \subseteq X}}^{-1}(A_U) \monoto F_X .
\end{gather*}
Then $\natjoin A$ is the limit of $A$ as a diagram in $\DD$.
\end{fact}

For each $U \in \C$, write $\rho_U : \natjoin A \to A_U$ for the restriction of $F_{U \subseteq X}$ to $\natjoin_A$;
it generalizes the restriction of global sections to local sections over $U$.
\autoref{def:contextuality.sets} then extends to

\begin{definition}\label{def:contextuality.general}
An $\DD$-valued separated presheaf $A$ is said to be logically contextual if not every $\rho_U : \natjoin A \to A_U$ is a regular epi.
$A$ is moreover said to be strongly contextual if $\natjoin A$ is not well-supported, i.e.\ if the unique arrow ${!}_{\natjoin A} : \natjoin A \to 1$ is not a regular epi.
\end{definition}

Rewriting this in the internal language of $\DD$, the strong contextuality of $A$ means that
$\DD$ fails $\exists \bar{x} : F_X \ldot \natjoin A(\bar{x})$, i.e., that no global section $\bar{x}$ satisfies all the constraints $A_U$.
The logical contextuality means that
\begin{gather*}
\bar{x} : F_V \mid A_V(\bar{x}) \vdash \exists \bar{y} : F_{X \setminus V} \ldot \natjoin A(\langle \bar{x}, \bar{y} \rangle)
\end{gather*}
fails in $\DD$ for some $V \in \C$,
i.e., that not every local section $\bar{x}$ over $V$ satisfying $A_V$ extends to a global section $\langle \bar{x}, \bar{y} \rangle$ satisfying all $A_U$.
(We should stress that the definition makes sense for any separated presheaves $A$ and not just no-signalling ones.)

\subsection{Contextual Interpretation}\label{sec:general.interpretation}

In \autoref{def:inchworm.logic} we defined a contextual language $\L_\C$ and logic $\vdash_\C$ simply as a global language $\L$ and logic $\vdash$ paired with their contextual fragments.
Our definition of an interpretation of them in regular categories goes in parallel.

\begin{definition}\label{def:inchworm.interpretation}
Given a contextual language $\L_\C = (\L, \Phi_\C)$, an interpretation of it in a regular category $\DD$ is simply an interpretation $\Scott{-}$ of $\L$ in $\DD$.
The images of $T_x$ and $\Phi_U$ then play special r\^oles:
For each $\bar{x} \in \C$, we have
\begin{itemize}
\item
$\Scott{T_{\bar{x}}} = \prod_{x \in \bar{x}} \Scott{T_x}$;
therefore $\Scott{T_{-}} : \C^\op \to \DD$ forms a sheaf by \autoref{thm:separated.presheaf.general}.
\item
Moreover, $\Scott{\, \bar{x} : T_{\bar{x}} \mid \varphi \,} \monoto \Scott{T_{\bar{x}}}$ for each $\varphi \in \Phi_{\bar{x}}$.
\end{itemize}
So we may write $(\Scott{-}, F)$ for the interpretation $\Scott{-}$, where $F$ is the sheaf $F : \bar{x} \mapsto \Scott{T_{\bar{x}}}$.
We may also write $\Scott{\varphi}_{\bar{x}} \monoto F_{\bar{x}}$ for $\Scott{\, \bar{x} : T_{\bar{x}} \mid \varphi \,}$.
\end{definition}

\begin{example}\label{eg:AvN.interpretation}
Expanding \autoref{eg:AvN.syntax}, take $\Scott{-}$ in $\Sets$ with $\Scott{T} = \2$ and the obvious $\Scott{0}$, $\Scott{1}$, and $\Scott{\oplus}$.
Then we have a sheaf $\Scott{T_{-}} : U \mapsto \2^U$ and, e.g.,
\begin{gather*}
\begin{tikzpicture}[x=25pt,y=25pt,baseline=(A.base)]
\node (A) [anchor=east,inner sep=0.25em] at (0,0) {$\Scott{\, x : T, y : T \mid x \oplus y = 0 \,}$};
\node (B) [anchor=west,inner sep=0.25em] at ($ (A.east) + (1.5,0) $) {$\Scott{T_{\{ x, y \}}} = \2 \times \2$};
\node (D) [anchor=west,inner sep=0.25em] at ($ (B.east) + (3,0) $) {$\2$};
\node (C) [transform canvas={yshift=-3.2pt},inner sep=0.25em] at ($ (B.east) + (1.5,0) $) {$1$};
\draw [>->] (A) -- (B);
\draw [transform canvas={yshift=3.2pt},->] (B) -- (D) node [pos=0.5,inner sep=2pt,above] {$\Scott{\oplus}$};
\draw [transform canvas={yshift=-3.2pt},->] (B) -- (C) node [pos=0.5,inner sep=2pt,below] {${!}$};
\draw [transform canvas={yshift=-3.2pt},->] (C) -- (D) node [pos=0.5,inner sep=2pt,below] {$\Scott{0}$};
\end{tikzpicture}
\end{gather*}
is an equalizer of $\Scott{\oplus}$ and $\Scott{0} \cmp {!}$.
\end{example}

An interpretation $\Scott{-}$ of $\L$ is said to model a sequent $\Gamma \vdash \varphi$ if some finite $\Delta \subseteq \Gamma$ has $\bigwedge_{\psi \in \Delta} \Scott{\psi}_U \leqslant \Scott{\varphi}_U$.
This makes sense whether $U \in \C$ or not.
Nevertheless, if $U \notin \C$, then $\bigwedge_{\psi \in \Delta} \Scott{\psi}_U \leqslant \Scott{\varphi}_U$ only means the global entailment and not the local one.
Take, e.g.,

\begin{example}\label{eg:AvN.model}
Expanding \autoref{eg:AvN.interpretation}, the presheaf model $A$ of the PR box, (\hyperref[pic:PR]{d}) of \autoref{fig:Hardy.PR}, is a subpresheaf of $\Scott{T_{-}}$ described by \eqref{eq:AvN.PR}:
E.g.\ $A_{\{ a_i, b_1 \}} = \Scott{\, a_1 \oplus b_1 = 0 \,}_{\{ a_1, b_1 \}} \monoto \2 \times \2$.
Then the global inconsistency, and strong contextuality $\Gamma \vdash \bot$ in particular, of the equations $\Gamma$ in \eqref{eq:AvN.PR} means
$\natjoin A = \bigcap_{\varphi \in \Gamma} \Scott{\varphi}_X \subseteq \Scott{\bot}_X = \varnothing$.
Yet $\Gamma$ is locally consistent, modelled by the PR box.
\end{example}

This is why, to model the inchworm logic of local inference, we need a presheaf on different contexts $U \in \C$, as opposed to an intersection in a single context $V \notin \C$, to the left of $\leqslant$.

\begin{definition}\label{def:inchworm.pre-model}
Suppose $(\Scott{-}, F)$ is an interpretation of a contextual language $\L_\C$.
Then let us say that a subpresheaf $A \monoto F$ is a \emph{pre-model} in $(\Scott{-}, F)$ of a formula $\varphi \in \Phi_U$ in a context $U \in \C$, and write $A \vDash_U \varphi$, to mean that $A_U \leqslant \Scott{\varphi}_U$.
\end{definition}

\begin{fact}\label{thm:subpresheaf.pre-model}
If $A \leqslant B$ for subpresheaves $A$ and $B$ of $F$, then $B \vDash_U \varphi$ implies $A \vDash_U \varphi$.
\end{fact}

Note, however, that this notion of pre-model is context-dependent and concerns formulas in contexts as opposed to formulas \textit{per se}.
When $U \subseteq V \in \C$ and $\varphi \in \Phi_U$, \autoref{thm:no-signalling.subpresheaf} yields
\begin{enumerate}
\setcounter{enumi}{\value{equation}}
\item\label{item:pre-model.up}
$A \vDash_U \varphi$ entails $A \vDash_V \varphi$, because
\begin{gather*}
\def\fCenter{\leqslant}
\AX$\exists_{F_{U \subseteq V}}(A_V) \fCenter A_U \leqslant \Scott{\varphi}_U$
\UI$A_V \fCenter {F_{U \subseteq V}}^{-1} \Scott{\varphi}_U = \Scott{\varphi}_V$
\DP
\end{gather*}
\item\label{item:pre-model.down}
Suppose $A$ is no-signalling.
Then $A \vDash_V \varphi$ entails $A \vDash_U \varphi$.
This is because
\begin{gather*}
\def\fCenter{\leqslant}
\AX$A_V \fCenter \Scott{\varphi}_V = {F_{U \subseteq V}}^{-1} \Scott{\varphi}_U$
\UI$A_U = \exists_{F_{U \subseteq V}}(A_V) \fCenter \Scott{\varphi}_U$
\DP
\end{gather*}
\setcounter{equation}{\value{enumi}}
\end{enumerate}
If $A$ is not no-signalling, \eqref{item:pre-model.down} may fail, and then inchworm inference fails.
E.g., in \eqref{eq:AvN.PR.Charlie.steps}, the first step purports to show that, if $A \vDash_U a_1 = b_1$ and $A \vDash_U a_1 = c$, then $A \vDash_U b_1 = c$ and so $A \vDash_{U \cap V} b_1 = c$;
but the ``and so'' step here requires \eqref{item:pre-model.down}.
In this sense, as mentioned in \autoref{sec:model.no-signalling}, no-signalling means the context-independent coherence of a presheaf as a model of formulas.
Therefore

\begin{definition}\label{def:inchworm.model}
A pre-model $A$ is called a (no-signalling) \emph{model} if it is no-signalling.
Then we say that $A$ is a model of a formula $\varphi \in \Phi_\C$, and write $A \vDash \varphi$, to mean that $A$ is a pre-model of $\varphi$ in any suitable context, i.e., that $A_U \leqslant \Scott{\varphi}_U$ for every $U \in \C$ such that $\varphi \in \Phi_U$.
\end{definition}

\begin{theorem}\label{thm:soundness.model}
Let $\Scott{-}$ be an interpretation of $\L_\C$ that models a theory $\vdash$ in $\L$.
Then the inchworm logic $\vdash_\C$ of $\vdash$ is sound with respect to the no-signalling models in $\Scott{-}$:
If $\Gamma \vdash_\C \varphi$, then $A \vDash \varphi$ for every no-signalling model $A$ of $\Gamma$ in $\Scott{-}$.
\end{theorem}

\subsection{The Inchworm and No-Signalling}\label{sec:theorem.no-signalling}

\autoref{sec:general.interpretation} primarily concerned how given presheaves modelled formulas.
We showed in particular that no-signalling validated inchworm inference.
Let us discuss, on the other hand, how the description by given formulas yields a model.
This shows the other direction of the connection between no-signalling and the inchworm, from the latter to the former.

We say a set $\Gamma \subseteq \Phi_\C$ of formulas of $\L_\C$ is \emph{$\C$-finite} if $\Gamma_U$ is finite for each $U \in \C$.

\begin{definition}
Let $(\Scott{-}, F)$ be an interpretation of $\L_\C$.
Given any $\C$-finite $\Gamma \subseteq \Phi_\C$, define $\MM[F]{\Gamma}$ as a family $(\MM[F]{\Gamma}_U = \bigwedge_{\varphi \in \Gamma_U} \Scott{\varphi}_U \monoto F_U)_{U \in \C}$ of subobjects of $F_U$.
\end{definition}

\begin{fact}\label{thm:description.pre-model}
$\MM[F]{\Gamma}$ is the largest subpresheaf $A$ of $F$ such that $A \vDash_U \Gamma_U$ for each $U \in \C$.
\end{fact}

$\MM[F]{\Gamma}$ generally fails to be no-signalling.
E.g., $A$ in (\hyperref[pic:signal]{e}) of \autoref{fig:no-signalling} is $\MM[F]{\Gamma}$ given by $\Gamma = \{ \varphi \}$ for $\varphi = (a_2 \wedge \lnot b_1) \vee (\lnot a_2 \wedge \lnot b_1)$;
since $\varphi$ cannot be in the context $\{ b_1 \}$, $\Gamma_{\{ b_1 \}} = \varnothing$ and $A_{\{ b_1 \}} = \2$.
Yet the description by $\Gamma$ sometimes manages to give a no-signalling $\MM[F]{\Gamma}$.

\begin{fact}\label{thm:description.model}
Let $(\Scott{-}, F)$ be an interpretation of $\L_\C$ that models a theory $\vdash$ in $\L$.
We say $\Gamma \subseteq \Phi_\C$ is \emph{inchworm-saturated} if $\Gamma_V \vdash \varphi$ implies $\Gamma_U \vdash \exists_{V \setminus U} \ldot \varphi$ whenever $U \subseteq V \in \C$ and $\varphi \in \Phi_V$.
Now, if a $\C$-finite $\Gamma$ is inchworm-saturated, then $\MM[F]{\Gamma}$ is no-signalling.
\end{fact}

When $\Gamma$ is inchworm-saturated, it may not be deductively closed, but the inchworm cannot bring a new piece of information $\psi$ to a context $U$ from another $V$, since $\psi$ follows from the information $\Gamma_U$ that $U$ already has.
\autoref{thm:description.model} means that, if $\Gamma$ is inchworm-saturated and if each $\Gamma_U$ finite and consistent (and has $\MM[F]{\Gamma}_U$ nonempty or well-supported), then $\Gamma$ is locally consistent, modelled by a no-signalling model $\MM[F]{\Gamma}$.
E.g., \eqref{eq:AvN.PR} is inchworm-saturated, with each context consistent, so it gives the PR box, (\hyperref[pic:PR]{d}) of \autoref{fig:Hardy.PR}, as $\MM[F]{\Gamma}$.

On the other hand, even when a description $\Gamma$ is not inchworm-saturated and $\MM[F]{\Gamma}$ fails to be no-signalling, the inchworm can carve out the ``no-signalling interior'' of $\MM[F]{\Gamma}$, if $\Gamma$ can be saturated in finite (or $\C$-finite) steps.

\begin{theorem}\label{thm:inchworm.saturation}
Let $(\Scott{-}, F)$ be an interpretation of $\L_\C$ that models a theory $\vdash$ in $\L$.
Given $\Gamma \subseteq \Phi_\C$, suppose there is a $\C$-finite and inchworm-saturated $\Delta \subseteq \Phi_\C$ such that $\Gamma \subseteq \Delta$ and $\Gamma \vdash_\C \varphi$ for all $\varphi \in \Delta$.
Then $\MM[F]{\Delta}$ is the largest no-signalling subpresheaf of $\MM[F]{\Gamma}$.
\end{theorem}

Take again the example from right after \autoref{thm:description.pre-model}:
$\MM[F]{\Gamma}$ of $\Gamma = \{ \varphi \}$ for $\varphi = (a_2 \wedge \lnot b_1) \vee (\lnot a_2 \wedge \lnot b_1)$ is $A$ in (\hyperref[pic:signal]{e}) of \autoref{fig:no-signalling} and fails to be no-signalling.
Yet $\varphi \vdash \lnot b_1$, so $\Gamma \vdash_\C \lnot b_1$, and $\Delta = \Gamma \cup \{ \lnot b_1 \}$ is inchworm-saturated, with $\lnot b_1 \in \Delta_{\{ b_1 \}}$.
Hence, by \autoref{thm:inchworm.saturation}, the inchworm carves out a no-signalling $\MM[F]{\Delta}$ by removing the red sections from (\hyperref[pic:signal]{e}).
Indeed, in many applications (e.g.\ all the examples in Sections \ref{sec:model} and \ref{sec:logic}), the theory $\vdash$ satisfies
\begin{enumerate}
\setcounter{enumi}{\value{equation}}
\item\label{item:finite.theory}
Given any $\Gamma \subseteq \Phi_\C$ (that may not be $\C$-finite), for each $U \in \C$ there is a finite $\Delta_U \subseteq \Gamma_U$ such that $\Delta_U \vdash \varphi$ for all $\varphi \in \Gamma_U$.
\setcounter{equation}{\value{enumi}}
\end{enumerate}
This guarantees the supposition of \autoref{thm:inchworm.saturation}:
Given any $\C$-finite $\Gamma \subseteq \Phi_\C$, take its $\vdash_\C$-deductive closure $\Gamma^\ast = \{\, \varphi \mid \Gamma \vdash_\C \varphi \,\}$ as $\Gamma$ in \eqref{item:finite.theory} and obtain $\Delta_U$;
then $\Delta = \Gamma \cup \bigcup_{U \in \C} \Delta_U$ is such as in \autoref{thm:inchworm.saturation}.
Therefore \autoref{thm:inchworm.saturation} applies and leads to a family of completeness results organized by \autoref{thm:completeness.transfer.model}, which transfers a completeness theorem of a global theory to its inchworm fragment.
It yields, e.g., \autoref{thm:completeness.regular.model}, since any (global) regular theory has a ``conservative model'' in a ``classifying category'' (e.g.\ \cite[Proposition 6.4]{but98}).

\begin{lemma}\label{thm:completeness.transfer.model}
Suppose that a theory $\vdash$ in $\L$ satisfies \eqref{item:finite.theory}, and that $\Scott{-}$ is a conservative model of $\vdash$, meaning that, for any $\Gamma \subseteq \Phi_\C$, $\bigwedge_{\psi \in \Delta} \Scott{\psi}_U \leqslant \Scott{\varphi}_U$ for some $\Delta \subseteq \Gamma$ if but also only if $\Gamma \vdash \varphi$.
Then $\Gamma \vdash_\C \varphi$ iff $A \vDash \varphi$ for every no-signalling model $A$ of $\Gamma$ in $\Scott{-}$.
\end{lemma}

\begin{theorem}\label{thm:completeness.regular.model}
Let $\vdash$ be a regular theory satisfying \eqref{item:finite.theory}.
Then, for any $\Gamma \subseteq \Phi_\C$, $\Gamma \vdash_\C \varphi$ iff $A \vDash \varphi$ for every no-signalling model $A$ of $\Gamma$ in every model $\Scott{-}$ of $\vdash$ in any regular category.
\end{theorem}

\subsection{Completion for Completeness}

Generally, the property \eqref{item:finite.theory} may fail and the inchworm saturation may not be attained in finite steps.

\begin{example}\label{eg:spiral}
In \autoref{fig:Hardy.PR}, replace each $A_x = \2$ with $\ZZ$, and let $\Gamma$ be the set of formulas
\begin{gather*}
a_1 = b_2 , \mspace{18mu}
b_1 = a_1 , \mspace{18mu}
a_2 = b_1 , \mspace{18mu}
b_2 = a_2 + 1 , \mspace{18mu}
b_2 > 0
\end{gather*}
in the obvious $\L$ and $\vdash$.
Then $\Gamma \vdash_\C a_1 > 0, b_1 > 0, a_2 > 0, b_2 > 1, \ldots, x > n$ for every $x \in X$ and $n \in \NN$, whereas $\Gamma \nvdash_\C \bot$ (although the empty presheaf is the only no-signalling model of $\Gamma$).
So there cannot be any such $\Delta$ as in \autoref{thm:inchworm.saturation}.
(Note that the topology of $\C$ is essential:
E.g., if we take $\C = \pw X$ instead, then $\Gamma \vdash_\C \bot$ by $\Gamma \vdash \bot$.)
\end{example}

Thus, even if $\Gamma$ is finite, the set $\{\, \Scott{\varphi}_U \mid \varphi \in {\Gamma^\ast}_U \,\}$ may have no minimum (though it is lowerbounded by $\exists_{F_{U \subseteq X}}(\bigwedge_{\psi \in \Gamma} \Scott{\psi}_X)$ if $X$ is finite);
then, in a regular category in general, $\bigwedge_{\varphi \in {\Gamma^\ast}_U} \Scott{\varphi}_U$ may not exist.
So, instead of the semilattice $\Sub_\DD(F_U)$ of subobjects, let us use a completion of it, viz.\ the semilattice $\Filt(\Sub_\DD(F_U))$ of filters in $\Sub_\DD(F_U)$,
and assign a filter of subobjects, instead of a subobject, to each $U \in \C$.

\begin{definition}\label{def:inchworm.filter.model}
Suppose $(\Scott{-}, F)$ is an interpretation of a contextual language $\L_\C$.
Then, by a \emph{filter model} in $(\Scott{-}, F)$, we mean a presheaf $G : \C^\op \to \Sets$ such that
\begin{itemize}
\item
$G_U \in \Filt(\Sub_\DD(F_U))$ for every $U \in \C$.
\item
For $U \subseteq V \in \C$,
\begin{gather*}
G_U = \{\, S \monoto F_U \mid {F_{U \subseteq V}}^{-1}(S) \in G_V \,\} = \{\, \exists_{F_{U \subseteq V}}(S) \monoto F_U \mid S \in G_V \,\} ,
\end{gather*}
so $G_{U \subseteq V} : G_V \to G_U :: S \mapsto \exists_{F_{U \subseteq V}}(S)$ is a surjection.
\end{itemize}
We say $G$ models $\varphi$, and write $G \vDash \varphi$, to mean that $\Scott{\varphi}_U \in G_U$ whenever $\varphi \in \Phi_U$.
\end{definition}

Then we have the filter versions of \autoref{thm:soundness.model}, \autoref{thm:description.model}, and completeness results organized by \autoref{thm:completeness.transfer.model}.
Observe that every (no-signalling) model $A$ is a ``principal'' filter model, $U \mapsto \upset{A_U} = \{\, S \monoto F \mid A_U \leqslant S \,\}$;
so \autoref{thm:soundness.filter} is stronger than \autoref{thm:soundness.model}.

\begin{theorem}\label{thm:soundness.filter}
Let $\Scott{-}$ be a model of a theory $\vdash$ in $\L$.
Then the inchworm logic $\vdash_\C$ of $\vdash$ is sound with respect to the filter models in $\Scott{-}$:
If $\Gamma \vdash_\C \varphi$, then $G \vDash \varphi$ for every filter model $G$ of $\Gamma$ in $\Scott{-}$.
\end{theorem}

\begin{fact}\label{thm:description.filter}
Let $(\Scott{-}, F)$ be a model of a theory $\vdash$ in $\L$.
Given any $\Gamma \subseteq \Phi_\C$, the family $\FiltMM[F]{\Gamma} = (\{\, S \monoto F_U \mid \Scott{\varphi}_U \leqslant S$ for some $\varphi \in {\Gamma^\ast}_U \,\})_{U \in \C}$ is a filter model of $\Gamma$ in $(\Scott{-}, F)$.
Moreover, for any filter model $G$ of $\Gamma$ in $(\Scott{-}, F)$, $\FiltMM[F]{\Gamma}_U \subseteq G_U$ for each $U \in \C$.
\end{fact}

\begin{lemma}\label{thm:completeness.transfer.filter}
Suppose $\Scott{-}$ is a conservative model of a theory $\vdash$ in $\L$.
Then $\Gamma \vdash_\C \varphi$ iff $G \vDash \varphi$ for every filter model $G$ of $\Gamma$ in $\Scott{-}$.
\end{lemma}

\section{Conclusion}

We have formulated contextual models as presheaves valued in regular categories, as well as providing ``inchworm logic'' for local inference in those contextual models.
We have also proven ``completeness-transfer lemmas'' (Lemmas \ref{thm:completeness.transfer.model} and \ref{thm:completeness.transfer.filter}), so that completeness theorems in categorical logic transfer straightforwardly to inchworm logic.

Let us conclude the paper by discussing connections and applications between the framework of this paper and other approaches or other fields as future work.
First of all, categorical logic has a long tradition (since \cite{law70}) of viewing local truth as a modal operator.
Indeed, the logic of local information in this paper is closely related to the dynamic-logical characterization of contextuality in \cite{kis14}.
There is also a connection to model theory.
For instance, the similarity between inchworm inference and Craig interpolation should be obvious;
indeed, by defining $\Phi_U$ more generally as a ``language in the vocabulary $U$'', we can prove a stronger version of Robinson's joint consistency theorem (see \cite[Subsection 4.1.1]{GM05}) that is sensitive to the topology of $\C$.

As explained in \autoref{sec:model}, presheaf models, whether no-signalling or not, can model Boolean valuations.
This enables us to transfer and apply techniques from satisfiability problems to quantum contextuality as computational resource.
Another connection is to the structure of valuation algebra, which is used for local computation \cite{KPS12}.
In fact, our presheaf models can also be formulated in terms of valuation algebras, as a $\C$-indexed family of valuations satisfying certain conditions.
We can expect these connections to help extend local computation to situations in classical computing where contextual phenomena arise.

The generality of taking presheaves in regular categories is also expected to facilitate applications.
In cohomology, it is typical to use presheaves valued in regular categories, such as presheaves of abelian groups, $R$-modules, etc.
Therefore the framework of this paper applies to the logic of local inference within such presheaves.
One can also take regular categories of structures that are used for other purposes such as modelling processes in quantum physics.
In addition, the connection to logical paradoxes \cite{ABKLM15} is also relevant.
As shown in \cite{law69,AZ15}, regular categories provide background for self-referential and other fixpoint paradoxes;
so our formalism will unify the two perspectives on paradoxes.

\appendix

\section{Proofs}

\setcounter{theorem}{0}

\begin{fact}
A presheaf $P$ on a simplicial complex $\C$ is a sheaf iff $P_U = \prod_{x \in U} P_x$ for all $U \in \C$.
$P$ is separated iff it is a subpresheaf of a sheaf, i.e.\ iff $P_U \subseteq \prod_{x \in U} P_x$ for all $U \in \C$.
\end{fact}

\begin{proof}
This is just an instance of \autoref{thm:separated.presheaf.general} with $\DD = \Sets$.
\end{proof}

\begin{fact}
$\sPsh(\C) \simeq \ndSimp / C$ for any simplicial complex $\C$.
\end{fact}

\begin{proof}
Let $\C$ be a simplicial complex on a set $X$.
Define functors
\begin{gather*}
\begin{tikzpicture}[x=25pt,y=25pt,baseline=(C.base)]
\node (C) [anchor=east,inner sep=0.25em] at (0,0) {$\sPsh(\C)$};
\node (D) [anchor=west,inner sep=0.25em] at (2.5,0) {$\ndSimp / C$};
\draw [transform canvas={yshift=2.5pt},->] (C) -- (D) node [pos=0.5,inner sep=2pt,above] {$F$};
\draw [transform canvas={yshift=-2.5pt},->] (D) -- (C) node [pos=0.5,inner sep=2pt,below] {$G$};
\end{tikzpicture}
\end{gather*}
as follows.
First we define $F$.
Let $P$ be a separated presheaf on $\C$.
By \autoref{thm:separated.presheaf}, $P_U \subseteq \prod_{x \in U} P_x$ for all $U \in \C$.
Then let $F P$ be a map $\pi : A \to X$ from a simplicial complex $\A$ on $A$ such that
\begin{itemize}
\item
$A = \sum_{x \in X} P_x$ and $\pi$ is the projection $\pi : A \to X :: (x, a) \mapsto x$.
\item
$\A = \bigcup_{U \in \C} P_U$.
That is, for each $S \subseteq A$, we have $S \in \A$ iff there is some $U \in \C$ such that $S \in P_U$, i.e.\ such that $S = \{\, (x, s(x)) \mid x \in U \,\}$ for some $s \in P_U \subseteq \prod_{x \in U} P_x$.
\end{itemize}
Then $F P = \pi : \A \to \C$ is clearly a non-degenerate bundle over $\C$.

Let $\vartheta : P \to Q$ be a natural transformation between separated presheaves $P$ and $Q$ on $\C$;
write $F P = \pi_1 : \A_1 \to \C$ and $F Q = \pi_2 : \A_2 \to \C$, with $\A_1$ and $\A_2$ on the sets $A_1 = \sum_{x \in X} P_x$ and $A_2 = \sum_{x \in X} Q_x$.
Then let $F \vartheta : A_1 \to A_2$ be the map such that
\begin{gather*}
F \vartheta : \sum_{x \in X} P_x \to \sum_{x \in X} Q_x :: (x, a) \mapsto (x, \vartheta_x(a)) .
\end{gather*}
$F \vartheta$ is a non-degenerate map of bundles from $F P$ to $F Q$ over $\C$:
\begin{itemize}
\item
If $s \in \A_1$, then $s \in P_U \subseteq \prod_{x \in U} P_x$ for some $U \in \C$, and so
\begin{gather*}
F \vartheta[s] = \{\, (x, \vartheta_x(s(x)) \mid x \in U \,\} = \{\, (x, (\vartheta_U(s))(x) \mid x \in U \,\} = \vartheta_U(s) \in \A_2
\end{gather*}
by the naturality of $\vartheta$.
\begin{align*}
\begin{tikzpicture}[x=25pt,y=25pt,baseline=(O.base)]
\coordinate (O) at (0,0);
\coordinate (r) at (2.5,0);
\coordinate (d) at (0,-2);
\node (Dom1) [inner sep=0.25em] at (O) {$P_U$};
\node (Dom2) [inner sep=0.25em] at ($ (Dom1) + (d) $) {$P_x$};
\node (Cod1) [inner sep=0.25em] at ($ (Dom1) + (r) $) {$Q_U$};
\node (Cod2) [inner sep=0.25em] at ($ (Dom2) + (r) $) {$Q_x$};
\draw [->] (Dom1) -- (Dom2);
\draw [->] (Cod1) -- (Cod2);
\draw [->] (Dom1) -- (Cod1) node [pos=0.5,inner sep=2pt,above] {$\vartheta_U$};
\draw [->] (Dom2) -- (Cod2) node [pos=0.5,inner sep=2pt,below] {$\vartheta_x$};
\coordinate (Dom2') at (Dom2);
\coordinate (Cod1') at (Cod1);
\path (Dom2') -- (Cod1') node [pos=0.5,sloped] {$=$};
\end{tikzpicture}
&&
\begin{tikzpicture}[x=25pt,y=25pt,baseline=(O.base)]
\coordinate (O) at (0,0);
\coordinate (r) at (4.75,0);
\coordinate (d) at (0,-2);
\node (Dom1) [inner sep=0.25em] at (O) {$s$};
\node (Dom2) [inner sep=0.25em] at ($ (Dom1) + (d) $) {$s(x)$};
\node (Cod1) [inner sep=0.25em] at ($ (Dom1) + (r) $) {$\vartheta_U(s)$};
\node (Cod2) [inner sep=0.25em] at ($ (Dom2) + (r) $) {$(\vartheta_U(s))(x)$};
\node (Cod2a) [anchor=east,inner sep=0.25em] at ($ (Cod2.west) + (0.5em,0) $) {$\vartheta_x(s(x)) = {}$};
\draw [|->] (Dom1) -- (Dom2);
\draw [|->] (Cod1) -- (Cod2);
\draw [|->] (Dom1) -- (Cod1);
\draw [|->] (Dom2) -- (Cod2a);
\end{tikzpicture}
\end{align*}
\item
$\pi_2 \cmp F \vartheta = \pi_1$ because $\pi_2 \cmp F \vartheta(x, a) = \pi_2(x, \vartheta_x(a)) = x = \pi_1(x, a)$.
\item
Because $\pi_1$ is non-degenerate and $\pi_2 \cmp F \vartheta = \pi_1$, $F \vartheta$ is non-degenerate as well.
\end{itemize}

Next we define $G$.
Given a non-degenerate bundle $\pi : \A \to \C$, for each $U \in \C$ let
\begin{gather*}
(G \pi)_U = P_U = \{\, s \in \A \mid \pi[s] = U \,\} .
\end{gather*}
Then $G \pi = P : U \mapsto P_U$ is a presheaf, with $P_{V \subseteq U} : P_U \to P_V :: s \mapsto s \cap \pi^{-1}[V]$.
Moreover $\{ P_{\{ x \} \subseteq U} \}_{x \in U}$ is jointly monic;
so $G \pi = P$ is a separated presheaf by \autoref{thm:separated.presheaf}.

Given a non-degenerate map of bundles $f : \pi_1 \to \pi_2$ between $\pi_i : \A_i \to \C$ ($i = 1, 2$), let $G f : G \pi_1 \to G \pi_2$ be the family $\{ (G f)_U \}_{U \in \C}$ of functions
\begin{gather*}
(G f)_U : (G \pi_1)_U \to (G \pi_2)_U :: s \mapsto f[s] ,
\end{gather*}
which lands in $(G \pi_2)_U$ because $f$ is a map over $\C$, i.e.\ $\pi_2 \cmp f[s] = \pi_1[s]$.
To show $G f$ to be a natural transformation, fix any $s \in (G \pi_1)_U$.
Then the non-degeneracy of $\pi_1$, $\pi_2$, and $f$, together with $\pi_2 \cmp f = \pi_1$, implies that the triangle on the left below commutes and that all the three arrows are bijections.
Therefore the diagram on the right, with $h = {\pi_2}^{-1}[-] \cap f[s]$, commutes as well:
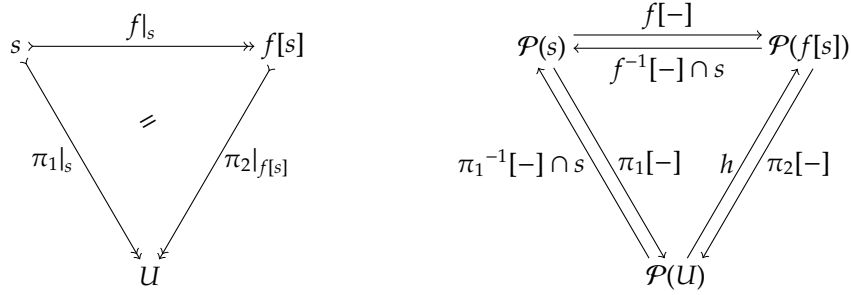
\begin{align*}
\begin{tikzpicture}[x=25pt,y=25pt,baseline=(C.base)]
\coordinate (O) at (0,0);
\node (A1) [inner sep=0.25em] at (O) {$s$};
\node (A2) [inner sep=0.25em] at ($ (A1) + (4,0) $) {$f[s]$};
\node (C) [inner sep=0.25em] at ($ (A1) + (2,-3.464) $) {$U$};
\draw [>->>] (A1) -- (A2) node [pos=0.5,inner sep=2pt,above] {$\rest{f}{s}$};
\draw [>->>] (A1) -- (C) node [pos=0.5,inner sep=3pt,left] {$\rest{\pi_1}{s}$};
\draw [>->>] (A2) -- (C) node [pos=0.5,inner sep=3pt,right] {$\rest{\pi_2}{f[s]}$};
\coordinate (A2') at (A2);
\coordinate (A1-C) at ($ (A1)!0.5!(C) $);
\path (A2') -- (A1-C) node [pos=0.67,sloped] {$=$};
\end{tikzpicture}
&&
\begin{tikzpicture}[x=25pt,y=25pt,baseline=(C.base)]
\coordinate (O) at (0,0);
\node (A1) [inner sep=0.25em] at (O) {$\pw(s)$};
\node (A2) [inner sep=0.25em] at ($ (A1) + (4,0) $) {$\pw(f[s])$};
\node (C) [inner sep=0.25em] at ($ (A1) + (2,-3.464) $) {$\pw(U)$};
\draw [transform canvas={yshift=5pt},->] (A1) -- (A2) node [pos=0.5,inner sep=2pt,above] {$f[-]$};
\draw [->] (A2) -- (A1) node [pos=0.5,inner sep=2pt,below] {$f^{-1}[-] \cap s$};
\draw [->] (A1) -- (C) node [pos=0.5,inner sep=3pt,right] {$\pi_1[-]$};
\draw [transform canvas={xshift=-5.774pt},->] (C) -- (A1) node [pos=0.5,inner sep=3pt,left] {${\pi_1}^{-1}[-] \cap s$};
\draw [transform canvas={xshift=5.774pt},->] (A2) -- (C) node [pos=0.5,inner sep=3pt,right] {$\pi_2[-]$};
\draw [->] (C) -- (A2) node [pos=0.5,inner sep=3pt,left] {$h$};
\end{tikzpicture}
\end{align*}
In particular, $h = f[{\pi_1}^{-1}[-] \cap s]$, which means that the square below commutes.
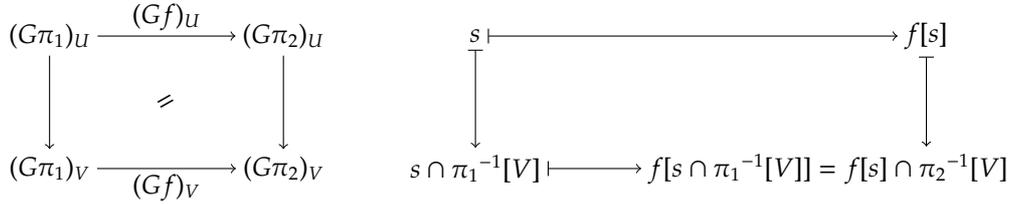
\begin{align*}
\begin{tikzpicture}[x=25pt,y=25pt,baseline=(O.base)]
\coordinate (O) at (0,0);
\coordinate (r) at (3.5,0);
\coordinate (d) at (0,-2);
\node (Dom1) [inner sep=0.25em] at (O) {$(G \pi_1)_U$};
\node (Dom2) [inner sep=0.25em] at ($ (Dom1) + (d) $) {$(G \pi_1)_V$};
\node (Cod1) [inner sep=0.25em] at ($ (Dom1) + (r) $) {$(G \pi_2)_U$};
\node (Cod2) [inner sep=0.25em] at ($ (Dom2) + (r) $) {$(G \pi_2)_V$};
\draw [->] (Dom1) -- (Dom2);
\draw [->] (Cod1) -- (Cod2);
\draw [->] (Dom1) -- (Cod1) node [pos=0.5,inner sep=2pt,above] {$(G f)_U$};
\draw [->] (Dom2) -- (Cod2) node [pos=0.5,inner sep=2pt,below] {$(G f)_V$};
\coordinate (Dom2') at (Dom2);
\coordinate (Cod1') at (Cod1);
\path (Dom2') -- (Cod1') node [pos=0.5,sloped] {$=$};
\end{tikzpicture}
&&
\begin{tikzpicture}[x=25pt,y=25pt,baseline=(O.base)]
\coordinate (O) at (0,0);
\coordinate (r) at (6.75,0);
\coordinate (d) at (0,-2);
\node (Dom1) [inner sep=0.25em] at (O) {$s$};
\node (Dom2) [inner sep=0.25em] at ($ (Dom1) + (d) $) {$s \cap {\pi_1}^{-1}[V]$};
\node (Cod1) [inner sep=0.25em] at ($ (Dom1) + (r) $) {$f[s]$};
\node (Cod2) [inner sep=0.25em] at ($ (Dom2) + (r) $) {$f[s] \cap {\pi_2}^{-1}[V]$};
\node (Cod2a) [anchor=east,inner sep=0.25em] at ($ (Cod2.west) + (0.5em,0) $) {$f[s \cap {\pi_1}^{-1}[V]] = {}$};
\draw [|->] (Dom1) -- (Dom2);
\draw [|->] (Cod1) -- (Cod2);
\draw [|->] (Dom1) -- (Cod1);
\draw [|->] (Dom2) -- (Cod2a);
\end{tikzpicture}
\end{align*}
Thus $G f$ is a natural transformation.

It is obvious that $G \cmp F \cong 1_{\sPsh(\C)}$ and $F \cmp G \cong 1_{\ndSimp / C}$ naturally.
\end{proof}

\addtocounter{theorem}{5}

\begin{fact}
An $\DD$-valued presheaf $P$ on a simplicial complex $\C$ is a sheaf iff $P_U = \prod_{x \in U} P_x$ for all $U \in \C$.
$P$ is separated iff it is a subpresheaf of a sheaf, i.e.\ iff each $\langle P_{\{ x \} \subseteq U} \rangle_{x \in U} : P_U \to \prod_{x \in U} P_x$ is monic, i.e.\ iff each $P_U$ is a relation in $\DD$ on $(P_x)_{x \in U}$.
\end{fact}

\begin{proof}
For the sheaf part, suppose $P$ is a sheaf.
First, $\bigcup_{i \in \varnothing} U_i = \varnothing$ gives an equalizer
\begin{gather*}
\begin{tikzpicture}[x=25pt,y=25pt,baseline=(A.base)]
\node (A) [anchor=east,inner sep=0.25em] at (0,0) {$P_\varnothing$};
\node (B) [anchor=west,inner sep=0.25em] at ($ (A.east) + (2,0) $) {$1$};
\node (C) [anchor=west,inner sep=0.25em] at ($ (B.east) + (2,0) $) {$1$};
\draw [>->] (A) -- (B);
\draw [transform canvas={yshift=3.2pt},->] (B) -- (C) node [pos=0.5,inner sep=2pt,above] {$1$};
\draw [transform canvas={yshift=-3.2pt},->] (B) -- (C) node [pos=0.5,inner sep=2pt,below] {$1$};
\end{tikzpicture}
\end{gather*}
that is an iso;
so $P_\varnothing = 1$.
Then for every $U \in \C$, $\bigcup_{x \in U} \{ x \} = U$ gives an equalizer
\begin{gather*}
\begin{tikzpicture}[x=25pt,y=25pt,baseline=(A.base)]
\node (A) [anchor=east,inner sep=0.25em] at (0,0) {$P_U$};
\node (B) [anchor=west,inner sep=0.25em] at ($ (A.east) + (3,0) $) {$\displaystyle{\prod_{x \in U} P_x}$};
\node (C) [anchor=west,inner sep=0.25em] at ($ (B.east) + (5,0) $) {$\displaystyle{\prod_{y, z \in U} P_{\{ y \} \cap \{ z \}}}$};
\draw [>->] (A) -- (B) node [pos=0.5,inner sep=2pt,above] {$\langle P_{\{ x \} \subseteq U} \rangle_{x \in U}$};
\draw [transform canvas={yshift=3.2pt},->] (B) -- (C) node [pos=0.5,inner sep=2pt,above] {$\langle P_{\{ y \} \cap \{ z \} \subseteq \{ y \}} \cmp p_y \rangle_{y, z \in U}$};
\draw [transform canvas={yshift=-3.2pt},->] (B) -- (C) node [pos=0.5,inner sep=2pt,below] {$\langle P_{\{ y \} \cap \{ z \} \subseteq \{ z \}} \cmp p_z \rangle_{y, z \in U}$};
\end{tikzpicture}
\end{gather*}
but then the parallel arrows are both isos, and therefore the equalizer is an iso, i.e.\ $P_U = \prod_x P_x$.
For the other direction, suppose $P_U = \prod_{x \in U} P_x$ for all $U \in \C$ and fix $\bigcup_i U_i = U$.
Then the equalizer we want is
\begin{gather*}
\begin{tikzpicture}[x=25pt,y=25pt,baseline=(A.base)]
\node (A) [anchor=east,inner sep=0.25em] at (0,0) {$\displaystyle{\prod_{x \in U} P_x}$};
\node (B) [anchor=west,inner sep=0.25em] at ($ (A.east) + (2.75,0) $) {$\displaystyle{\prod_i \prod_{x \in U_i} P_x}$};
\node (C) [anchor=west,inner sep=0.25em] at ($ (B.east) + (5.25,0) $) {$\displaystyle{\prod_{j, k} \prod_{x \in U_j \cap U_k} P_x}$};
\draw [>->] (A) -- (B) node [pos=0.5,inner sep=2pt,above] {$\langle \langle p_x \rangle_{x \in U_i} \rangle_i$};
\draw [transform canvas={yshift=3.2pt},->] (B) -- (C) node [pos=0.5,inner sep=2pt,above] {$g = \langle \langle p_x \rangle_{x \in U_j \cap U_k} \cmp p_j \rangle_{j, k}$};
\draw [transform canvas={yshift=-3.2pt},->] (B) -- (C) node [pos=0.5,inner sep=2pt,below] {$h = \langle \langle p_x \rangle_{x \in U_j \cap U_k} \cmp p_k \rangle_{j, k}$};
\end{tikzpicture}
\end{gather*}
But it is easy to check that any $f$ such that $g \cmp f = h \cmp f$ has the form $\langle \langle f_x \rangle_{x \in U_i} \rangle_i$, and hence factors uniquely as $\langle \langle p_x \rangle_{x \in U_i} \rangle_i \cmp \langle f_x \rangle_{x \in U}$.
Therefore $\langle \langle p_x \rangle_{x \in U_i} \rangle_i$ is an equalizer as desired.

For the presheaf part, fix $P$ and number the conditions as follows:
\begin{enumerate}
\def\theenumi{\roman{enumi}}
\item\label{item:separated.presheaf.general.1}
$P$ is separated, i.e.\ $\langle P_{U_i \subseteq U} \rangle_i : P_U \to \prod_i P_{U_i}$ is monic whenever $\bigcup_i U_i = U$.
\item\label{item:separated.presheaf.general.2}
$P$ is a subpresheaf of a sheaf.
\item\label{item:separated.presheaf.general.3}
Each $\langle P_{\{ x \} \subseteq U} \rangle_{x \in U} : P_U \to \prod_{x \in U} P_x$ is monic.
\item\label{item:separated.presheaf.general.4}
Each $P_U$ is a relation in $\DD$ on $(P_x)_{x \in U}$.
\end{enumerate}
\eqref{item:separated.presheaf.general.4} is defined by \eqref{item:separated.presheaf.general.3}.
\eqref{item:separated.presheaf.general.1} implies \eqref{item:separated.presheaf.general.3} as an instance.
\eqref{item:separated.presheaf.general.3} implies \eqref{item:separated.presheaf.general.2}, since $U \mapsto \prod_{x \in U} P_x$ is a sheaf by the sheaf part above.
So suppose \eqref{item:separated.presheaf.general.2}, that $m : P \monoto F$ for some sheaf $F$.
This means that the following square commutes, in which $\langle F_{U_i \subseteq U} \rangle_i$ is an equalizer and so monic.
\begin{gather*}
\begin{tikzpicture}[x=25pt,y=25pt,baseline=(O.base)]
\coordinate (O) at (0,0);
\coordinate (r) at (4,0);
\coordinate (d) at (0,-2.25);
\node (Dom1) [inner sep=0.25em] at (O) {$P_U$};
\node (Dom2) [inner sep=0.25em] at ($ (Dom1) + (d) $) {$F_U$};
\node (Cod1) [inner sep=0.25em] at ($ (Dom1) + (r) $) {$\displaystyle{\prod_i P_{U_i}}$};
\node (Cod2) [inner sep=0.25em] at ($ (Dom2) + (r) $) {$\displaystyle{\prod_i F_{U_i}}$};
\draw [>->] (Dom1) -- (Dom2) node [pos=0.5,inner sep=2pt,left] {$m_U$};
\draw [>->] (Cod1) -- (Cod2) node [pos=0.5,inner sep=2pt,right] {$\langle m_{U_i} \rangle_i$};
\draw [->] (Dom1) -- (Cod1) node [pos=0.5,inner sep=2pt,above] {$\langle P_{U_i \subseteq U} \rangle_i$};
\draw [>->] (Dom2) -- (Cod2) node [pos=0.5,inner sep=2pt,below] {$\langle F_{U_i \subseteq U} \rangle_i$};
\coordinate (Dom2') at (Dom2);
\coordinate (Cod1') at (Cod1);
\path (Dom2') -- (Cod1') node [pos=0.5,sloped] {$=$};
\end{tikzpicture}
\end{gather*}
Therefore $\langle P_{U_i \subseteq U} \rangle_i$ is monic, too.
Thus \eqref{item:separated.presheaf.general.1}.
\end{proof}

\addtocounter{theorem}{1}

\begin{fact}
Let $F$ be a sheaf on a simplicial complex $\C$.
Then a family $(i_U : A_U \monoto F_U)_{U \in \C}$ of subobjects forms a subpresheaf of $F$, and hence a separated presheaf, iff $A_V \leqslant {F_{U \subseteq V}}^{-1}(A_U)$, or equivalently $\exists_{F_{U \subseteq V}}(A_V) \leqslant A_U$, whenever $U \subseteq V \in \C$.
Moreover, a separated presheaf $i : A \monoto F$ is no-signalling iff $\exists_{F_{U \subseteq V}}(A_V) = A_U$ whenever $U \subseteq V \in \C$.
\end{fact}

\begin{proof}
Fix $U \subseteq V \in \C$ and write $f = F_{U \subseteq V}$.
Then, in the following diagram, the longer dotted arrow making the outer square commute, i.e.\ $A_{U \subseteq V}$, exists iff the shorter one exists, i.e.\ iff $A_V \leqslant f^{-1}(A_U)$.
Hence the first ``iff'' in the fact.
\begin{gather*}
\begin{tikzpicture}[x=25pt,y=25pt,baseline=(O.base)]
\coordinate (O) at (0,0);
\coordinate (r) at (3,0);
\coordinate (d) at (0,-2);
\node (Inv) [inner sep=0.25em] at (O) {$f^{-1}(A_U)$};
\coordinate (Inv') at (Inv);
\node (AV) [inner sep=0.25em] at ($ (Inv) + (-1.5,1) $) {$A_V$};
\node (AU) [inner sep=0.25em] at ($ (Inv) + (r) $) {$A_U$};
\node (FV) [inner sep=0.25em] at ($ (Inv) + (d) $) {$F_V$};
\node (FU) [inner sep=0.25em] at ($ (FV) + (r) $) {$F_U$};
\draw [->] (Inv) -- (AU);
\draw [->] (FV) -- (FU) node [pos=0.5,inner sep=2pt,below] {$f$};
\draw [>->] (Inv) -- (FV);
\draw [>->] (AU) -- (FU) node [pos=0.5,inner sep=2pt,right] {$i_U$};
\coordinate (Inv-pb) at ($ (Inv) + (0.7,-0.7) $);
\draw ($ (Inv-pb) + (-0.3,0) $) -- (Inv-pb) -- ($ (Inv-pb) + (0,0.3) $);
\draw [->,dotted] (AV) to [bend left=12.5] coordinate [pos=0.5] (AV-AU) node [pos=0.5,inner sep=2pt,above] {$A_{U \subseteq V}$} (AU);
\draw [>->] (AV) to [bend right=25] coordinate [pos=0.5] (AV-FV) node [pos=0.5,inner sep=2pt,left] {$i_V$} (FV);
\draw [>->,dotted] (AV) -- (Inv);
\path (Inv') -- (AV-AU) node [pos=0.6,sloped] {$=$};
\path (Inv') -- (AV-FV) node [pos=0.6,yshift=-7.5pt,sloped] {$=$};
\end{tikzpicture}
\end{gather*}
For the second ``iff'', factorize $A_{U \subseteq V}$ into a regular epi $e$ followed by a mono $m$ as follows;
it gives $\exists_f(A_V)$ because $(i_U \cmp m) \cmp e$ is a regular epi-mono factorization of $f \cmp i_V$.
\begin{gather*}
\begin{tikzpicture}[x=25pt,y=25pt,baseline=(O.base)]
\coordinate (O) at (0,0);
\coordinate (r) at (3.5,0);
\coordinate (d) at (0,-2);
\node (AV) [inner sep=0.25em] at (O) {$A_V$};
\node (Dir) [inner sep=0.25em] at ($ (AV) + (r) + (-1,0) $) {$\exists_f(A_V)$};
\coordinate (Dir') at (Dir);
\node (AU) [inner sep=0.25em] at ($ (AV) + (r) + (1,0) $) {$A_U$};
\coordinate (AU') at (AU);
\node (FV) [inner sep=0.25em] at ($ (AV) + (d) $) {$F_V$};
\coordinate (FV') at (FV);
\node (FU) [inner sep=0.25em] at ($ (FV) + (r) $) {$F_U$};
\draw [-latex] (AV) -- (Dir) node [pos=0.5,inner sep=2pt,above] {$e$};
\draw [>->] (Dir) -- (AU) node [pos=0.5,inner sep=2pt,above] {$m$};
\draw [->] (FV) -- (FU) node [pos=0.5,inner sep=2pt,below] {$f$};
\draw [>->] (AV) -- (FV) node [pos=0.5,inner sep=2pt,left] {$i_V$};
\draw [>->] (AU) -- (FU) node [pos=0.5,inner sep=3pt,right] {$i_U$};
\draw [>->,dotted] (Dir) -- (FU) coordinate [pos=0.5] (Dir-FU);
\path (Dir') -- (FV') node [pos=0.5,sloped] {$=$};
\path (AU') -- (Dir-FU) node [pos=0.67,sloped] {$=$};
\end{tikzpicture}
\end{gather*}
Then, because a regular epi-mono factorization is unique, $A_{U \subseteq V} = m \cmp e$ is a regular epi iff $m$ is an iso.
\end{proof}

\begin{fact}
(Assume that $X$ is finite or that $\DD$ is complete.)
Given any $\DD$-valued separated presheaf $A$, let $F$ be a sheaf such that $i : A \monoto F$ and, using $A_U$ as predicates in the internal language of $\DD$, define
\begin{gather*}
\natjoin A = \Scott{\, \bar{x} : F_X \mid \bigwedge_{U \in \C} A_U(F_{U \subseteq X} \bar{x}) \,} = \bigwedge_{U \in \C} {F_{U \subseteq X}}^{-1}(A_U) \monoto F_X .
\end{gather*}
Then $\natjoin A$ is the limit of $A$ as a diagram in $\DD$.
\end{fact}

\begin{proof}
Recall that the sheaf $F$ extends to $\pw X$.
So, for every $U \in \C$, write $p_U = F_{U \subseteq X} : F_X \to F_U$, and define an arrow $\rho_U : \natjoin A \to A_U$ as the composition of the two top arrows in the following.
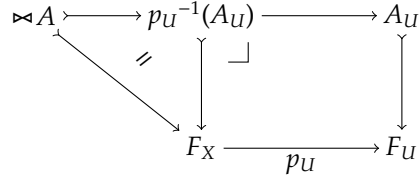
\begin{gather*}
\begin{tikzpicture}[x=25pt,y=25pt,baseline=(O.base)]
\coordinate (O) at (0,0);
\coordinate (r) at (3,0);
\coordinate (d) at (0,-2);
\node (Inv) [inner sep=0.25em] at (O) {${p_U}^{-1}(A_U)$};
\coordinate (Inv') at (Inv);
\node (Nat) [inner sep=0.25em] at ($ (Inv) + (-2.5,0) $) {$\natjoin A$};
\node (AU) [inner sep=0.25em] at ($ (Inv) + (r) $) {$A_U$};
\node (FX) [inner sep=0.25em] at ($ (Inv) + (d) $) {$F_X$};
\node (FU) [inner sep=0.25em] at ($ (FX) + (r) $) {$F_U$};
\draw [->] (FX) -- (FU) node [pos=0.5,inner sep=2pt,below] {$p_U$};
\draw [>->] (AU) -- (FU);
\draw [>->] (Inv) -- (FX);
\draw [->] (Inv) -- (AU);
\coordinate (Inv-pb) at ($ (Inv) + (0.7,-0.7) $);
\draw ($ (Inv-pb) + (-0.3,0) $) -- (Inv-pb) -- ($ (Inv-pb) + (0,0.3) $);
\draw [>->] (Nat) -- (Inv);
\draw [>->] (Nat) to coordinate [pos=0.5] (Nat-FX) (FX);
\path (Inv') -- (Nat-FX) node [pos=0.67,sloped] {$=$};
\end{tikzpicture}
\end{gather*}
Then $\natjoin A$ is a cone of the diagram $A$ because, for every $U \subseteq V \in \C$, we have $A_{U \subseteq V} \cmp \rho_V = \rho_U$ by chasing the following commutative diagram, where $f = F_{U \subseteq V}$.
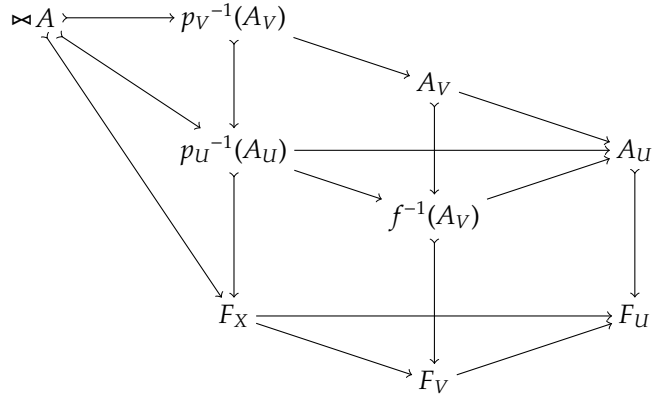
\begin{gather*}
\begin{tikzpicture}[x=25pt,y=25pt,baseline=(O.base)]
\coordinate (O) at (0,0);
\coordinate (r) at (3,0);
\coordinate (d) at (0,-1);
\coordinate (dd) at (0,-2);
\coordinate (ddd) at (0,-2.5);
\node (Nat) [inner sep=0.25em] at (O) {$\natjoin A$};
\node (pAV) [inner sep=0.25em] at ($ (Nat) + (r) $) {${p_V}^{-1}(A_V)$};
\node (pAU) [inner sep=0.25em] at ($ (pAV) + (dd) $) {${p_U}^{-1}(A_U)$};
\node (FX) [inner sep=0.25em] at ($ (pAU) + (ddd) $) {$F_X$};
\node (AV) [inner sep=0.25em] at ($ (pAV) + (r) + (d) $) {$A_V$};
\node (fAU) [inner sep=0.25em] at ($ (AV) + (dd) $) {$f^{-1}(A_V)$};
\node (FV) [inner sep=0.25em] at ($ (fAU) + (ddd) $) {$F_V$};
\node (AU) [inner sep=0.25em] at ($ (AV) + (r) + (d) $) {$A_U$};
\node (FU) [inner sep=0.25em] at ($ (AU) + (ddd) $) {$F_U$};
\draw [>->] (Nat) -- (pAV);
\draw [>->] (Nat) -- (pAU);
\draw [>->] (Nat) -- (FX);
\draw [>->] (pAV) -- (pAU);
\draw [>->] (pAU) -- (FX);
\draw [->] (pAV) -- (AV);
\draw [->] (pAU) -- (fAU);
\draw [->] (FX) -- (FV);
\draw [>->] (AV) -- (fAU);
\draw [>->] (fAU) -- (FV);
\draw [->] (pAU) -- (AU);
\draw [->] (FX) -- (FU);
\draw [->] (AV) -- (AU);
\draw [->] (fAU) -- (AU);
\draw [->] (FV) -- (FU);
\draw [>->] (AU) -- (FU);
\end{tikzpicture}
\end{gather*}
Now we claim that $\natjoin A$ is the terminal among the cones of $A$.
Fix any cone $C$ of $A$ with $g_U : C \to A_U$ for all $U \in \C$.
Then, by \autoref{thm:separated.presheaf.general}, for every $U \subseteq X$ we have $\langle i_x \cmp g_x \rangle_{x \in U} : C \to F_U$.
Therefore, for every $U \in \C$, the outer square below commutes and hence the arrow making the triangles commute, $u_U$, exists uniquely:
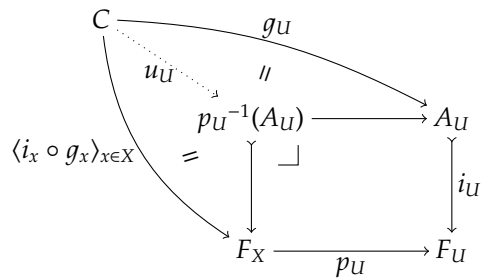
\begin{gather*}
\begin{tikzpicture}[x=25pt,y=25pt,baseline=(O.base)]
\coordinate (O) at (0,0);
\coordinate (r) at (3,0);
\coordinate (d) at (0,-2);
\node (Inv) [inner sep=0.25em] at (O) {${p_U}^{-1}(A_U)$};
\coordinate (Inv') at (Inv);
\node (AV) [inner sep=0.25em] at ($ (Inv) + (-2.25,1.5) $) {$C$};
\node (AU) [inner sep=0.25em] at ($ (Inv) + (r) $) {$A_U$};
\node (FV) [inner sep=0.25em] at ($ (Inv) + (d) $) {$F_X$};
\node (FU) [inner sep=0.25em] at ($ (FV) + (r) $) {$F_U$};
\draw [->] (Inv) -- (AU);
\draw [->] (FV) -- (FU) node [pos=0.5,inner sep=2pt,below] {$p_U$};
\draw [>->] (Inv) -- (FV);
\draw [>->] (AU) -- (FU) node [pos=0.5,inner sep=2pt,right] {$i_U$};
\coordinate (Inv-pb) at ($ (Inv) + (0.7,-0.7) $);
\draw ($ (Inv-pb) + (-0.3,0) $) -- (Inv-pb) -- ($ (Inv-pb) + (0,0.3) $);
\draw [->] (AV) to [bend left=12.5] coordinate [pos=0.5] (AV-AU) node [pos=0.5,inner sep=2pt,above] {$g_U$} (AU);
\draw [->] (AV) to [bend right=25] coordinate [pos=0.5] (AV-FV) node [pos=0.5,inner sep=2pt,left] {$\langle i_x \cmp g_x \rangle_{x \in X}$} (FV);
\draw [->,dotted] (AV) -- (Inv) node [pos=0.66,inner sep=3pt,left] {$u_U$};
\path (Inv') -- (AV-AU) node [pos=0.6,sloped] {$=$};
\path (Inv') -- (AV-FV) node [pos=0.6,yshift=-7.5pt,sloped] {$=$};
\end{tikzpicture}
\end{gather*}
Then the definition $\natjoin A = \bigwedge_{U \in \C} {p_U}^{-1}(A_U)$ means that there is a unique $v : C \to \natjoin A$ that makes the following diagram commute for each $U \subseteq V \in \C$.
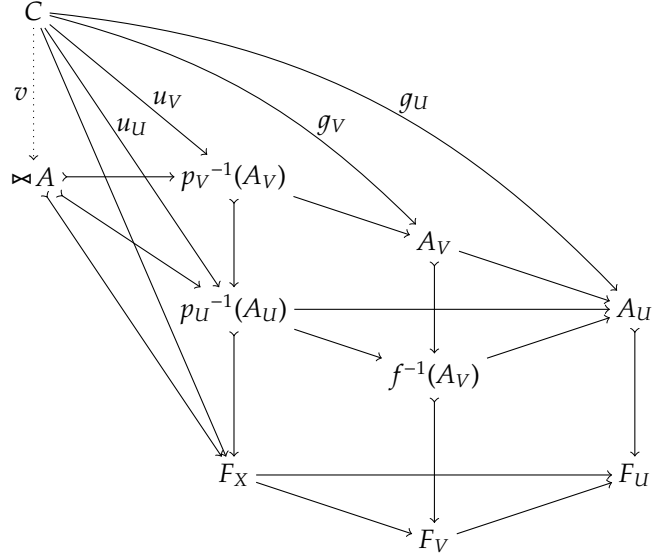
\begin{gather*}
\begin{tikzpicture}[x=25pt,y=25pt,baseline=(FV.base)]
\coordinate (O) at (0,0);
\coordinate (r) at (3,0);
\coordinate (d) at (0,-1);
\coordinate (dd) at (0,-2);
\coordinate (ddd) at (0,-2.5);
\node (Nat) [inner sep=0.25em] at (O) {$\natjoin A$};
\node (pAV) [inner sep=0.25em] at ($ (Nat) + (r) $) {${p_V}^{-1}(A_V)$};
\node (pAU) [inner sep=0.25em] at ($ (pAV) + (dd) $) {${p_U}^{-1}(A_U)$};
\node (FX) [inner sep=0.25em] at ($ (pAU) + (ddd) $) {$F_X$};
\node (AV) [inner sep=0.25em] at ($ (pAV) + (r) + (d) $) {$A_V$};
\node (fAU) [inner sep=0.25em] at ($ (AV) + (dd) $) {$f^{-1}(A_V)$};
\node (FV) [inner sep=0.25em] at ($ (fAU) + (ddd) $) {$F_V$};
\node (AU) [inner sep=0.25em] at ($ (AV) + (r) + (d) $) {$A_U$};
\node (FU) [inner sep=0.25em] at ($ (AU) + (ddd) $) {$F_U$};
\draw [>->] (Nat) -- (pAV);
\draw [>->] (Nat) -- (pAU);
\draw [>->] (Nat) -- (FX);
\draw [>->] (pAV) -- (pAU);
\draw [>->] (pAU) -- (FX);
\draw [->] (pAV) -- (AV);
\draw [->] (pAU) -- (fAU);
\draw [->] (FX) -- (FV);
\draw [>->] (AV) -- (fAU);
\draw [>->] (fAU) -- (FV);
\draw [->] (pAU) -- (AU);
\draw [->] (FX) -- (FU);
\draw [->] (AV) -- (AU);
\draw [->] (fAU) -- (AU);
\draw [->] (FV) -- (FU);
\draw [>->] (AU) -- (FU);
\node (C) [inner sep=0.25em] at ($ (Nat) + (0,2.5) $) {$C$};
\draw [->,dotted] (C) -- (Nat) node [pos=0.5,inner sep=2pt,left] {$v$};
\draw [->] (C) -- (pAV) node [pos=0.75,inner sep=5pt,above] {$u_V$};
\draw [->] (C) -- (pAU) node [pos=0.5,inner sep=8pt,above] {$u_U$};
\draw [->] (C) -- (FX);
\draw [->] (C) to [bend left=15] node [pos=0.75,inner sep=5pt,above] {$g_V$} (AV);
\draw [->] (C) to [bend left=20] node [pos=0.6,inner sep=4pt,above] {$g_U$} (AU);
\end{tikzpicture}
\qedhere
\end{gather*}
\end{proof}

\addtocounter{theorem}{5}

\begin{fact}
If $A \leqslant B$ for subpresheaves $A$ and $B$ of $F$, then $B \vDash_U \varphi$ implies $A \vDash_U \varphi$.
\end{fact}

\begin{proof}
Immediate by \autoref{def:inchworm.pre-model} of $\vDash_U$.
\end{proof}

\addtocounter{theorem}{1}

\begin{theorem}
Let $\Scott{-}$ be an interpretation of $\L_\C$ that models a theory $\vdash$ in $\L$.
Then the inchworm logic $\vdash_\C$ of $\vdash$ is sound with respect to the no-signalling models in $\Scott{-}$:
If $\Gamma \vdash_\C \varphi$, then $A \vDash \varphi$ for every no-signalling model $A$ of $\Gamma$ in $\Scott{-}$.
\end{theorem}

\begin{proof}
Immediate from \autoref{thm:soundness.filter}, since every model is a filter model.
\end{proof}

\addtocounter{theorem}{1}

\begin{fact}
$\MM[F]{\Gamma}$ is the largest subpresheaf $A$ of $F$ such that $A \vDash_U \Gamma_U$ for each $U \in \C$.
\end{fact}

\begin{proof}
Let $U \subseteq V \in \C$.
Then,
because $\Gamma_U \subseteq \Gamma_V$,
because $\Scott{\varphi}_V = {F_{U \subseteq V}}^{-1} \Scott{\varphi}_U$ for each $\varphi \in \Phi_U \subseteq \Phi_V$,
and
because the right adjoint ${F_{U \subseteq V}}^{-1}$ preserves meets,
we have
\begin{align*}
\MM[F]{\Gamma}_V
= \bigwedge_{\varphi \in \Gamma_V} \Scott{\varphi}_V
\leqslant \bigwedge_{\varphi \in \Gamma_U} \Scott{\varphi}_V
& = \bigwedge_{\varphi \in \Gamma_U} {F_{U \subseteq V}}^{-1} \Scott{\varphi}_U \\
& = {F_{U \subseteq V}}^{-1}(\bigwedge_{\varphi \in \Gamma_U} \Scott{\varphi}_U)
= {F_{U \subseteq V}}^{-1}(\MM[F]{\Gamma}_U) .
\end{align*}
Thus $\MM[F]{\Gamma}$ is a subpresheaf of $F$ by \autoref{thm:no-signalling.subpresheaf}, and $\varphi \in \Gamma_U$ implies $\MM[F]{\Gamma}_U = \bigwedge_{\psi \in \Gamma_U} \Scott{\psi}_U \leqslant \Scott{\varphi}_U$, i.e.\ $\MM[F]{\Gamma} \vDash_U \varphi$.
Now, let $A$ be a subpresheaf of $F$ such that $A \vDash_U \Gamma_U$ for every $U \in \C$.
Then, for every $U \in \C$ it has $A_U \leqslant \bigwedge_{\varphi \in \Gamma_U} \Scott{\varphi}_U = \MM[F]{\Gamma}_U$.
Thus $A \leqslant \MM[F]{\Gamma}$.
\end{proof}

\begin{fact}
Let $(\Scott{-}, F)$ be an interpretation of $\L_\C$ that models a theory $\vdash$ in $\L$.
We say $\Gamma \subseteq \Phi_\C$ is \emph{inchworm-saturated} if $\Gamma_V \vdash \varphi$ implies $\Gamma_U \vdash \exists_{V \setminus U} \ldot \varphi$ whenever $U \subseteq V \in \C$ and $\varphi \in \Phi_V$.
Now, if a $\C$-finite $\Gamma$ is inchworm-saturated, then $\MM[F]{\Gamma}$ is no-signalling.
\end{fact}

\begin{proof}
Fix any $U \subseteq V \in \C$.
Then the inchworm saturation of $\C$-finite $\Gamma$ implies that, since $\Gamma_V \vdash \bigwedge \Gamma_V$, we have $\Gamma_U \vdash \exists_{V \setminus U} \ldot \bigwedge \Gamma_V$.
Therefore a model $\Scott{-}$ of $\vdash$ has
\begin{gather*}
\MM[F]{\Gamma}_U = \bigwedge_{\varphi \in \Gamma_U} \Scott{\varphi}_U \leqslant \Scott{\exists_{V \setminus U} \ldot \bigwedge \Gamma_V}_U = \exists_{F_{U \subseteq V}}(\MM[F]{\Gamma}_V) .
\end{gather*}
Thus $\MM[F]{\Gamma}$ is no-signalling by \autoref{thm:no-signalling.subpresheaf}.
\end{proof}

\begin{theorem}
Let $(\Scott{-}, F)$ be an interpretation of $\L_\C$ that models a theory $\vdash$ in $\L$.
Given $\Gamma \subseteq \Phi_\C$, suppose there is a $\C$-finite and inchworm-saturated $\Delta \subseteq \Phi_\C$ such that $\Gamma \subseteq \Delta$ and $\Gamma \vdash_\C \varphi$ for all $\varphi \in \Delta$.
Then $\MM[F]{\Delta}$ is the largest no-signalling subpresheaf of $\MM[F]{\Gamma}$.
\end{theorem}

\begin{proof}
$\MM[F]{\Delta}$ is no-signalling by \autoref{thm:description.model}.
Hence \autoref{thm:description.pre-model} implies $\MM[F]{\Delta} \vDash \Gamma \subseteq \Delta$ and hence that $\MM[F]{\Delta}$ is a subpresheaf of $\MM[F]{\Gamma}$.
Now fix any no-signalling subpresheaf $A$ of $\MM[F]{\Gamma}$.
Then $A \vDash \Gamma$ by Facts \ref{thm:subpresheaf.pre-model} and \ref{thm:description.pre-model}.
This implies $A \vDash \Delta$ by soundness (\autoref{thm:soundness.model}), since $\Gamma \vdash_\C \varphi$ for all $\varphi \in \Delta$.
Therefore $A \leqslant \MM[F]{\Delta}$ by \autoref{thm:description.pre-model}.
\end{proof}

\begin{lemma}
Suppose that a theory $\vdash$ in $\L$ satisfies \eqref{item:finite.theory}, and that $\Scott{-}$ is a conservative model of $\vdash$, meaning that, for any $\Gamma \subseteq \Phi_\C$, $\bigwedge_{\psi \in \Delta} \Scott{\psi}_U \leqslant \Scott{\varphi}_U$ for some $\Delta \subseteq \Gamma$ if but also only if $\Gamma \vdash \varphi$.
Then $\Gamma \vdash_\C \varphi$ iff $A \vDash \varphi$ for every no-signalling model $A$ of $\Gamma$ in $\Scott{-}$.
\end{lemma}

\begin{proof}
The ``only if'' part is soundness (\autoref{thm:soundness.model}).
So, for the ``if'' part, suppose $A \vDash \varphi$ for every no-signalling model $A$ of $\Gamma$ in $\Scott{-}$.
By \eqref{item:finite.theory}, \autoref{thm:inchworm.saturation} applies and yields $\MM[F]{\Delta}$, which models $\Gamma \subseteq \Delta$.
Hence $\FiltMM[F]{\Delta} \vDash \varphi$.
This means that, for $U \in \C$ such that $\varphi \in \Phi_U$, $\bigwedge_{\psi \in \Delta_U} \Scott{\psi}_U = \FiltMM[F]{\Delta}_U \leqslant \Scott{\varphi}_U$.
Therefore $\Delta_U \vdash \varphi$ by the conservativity of $\Scott{-}$ and so $\Delta_U \vdash_\C \varphi$.
Then, since $\Gamma \vdash_\C \psi$ for all $\psi \in \Delta_U$, we have $\Gamma \vdash_\C \varphi$.
\end{proof}

\begin{theorem}
Let $\vdash$ be a regular theory satisfying \eqref{item:finite.theory}.
Then, for any $\Gamma \subseteq \Phi_\C$, $\Gamma \vdash_\C \varphi$ iff $A \vDash \varphi$ for every no-signalling model $A$ of $\Gamma$ in every model $\Scott{-}$ of $\vdash$ in any regular category.
\end{theorem}

\begin{proof}
The ``only if'' part is soundness (\autoref{thm:soundness.model}).
The ``if'' part is by \autoref{thm:completeness.transfer.model}, since $\vdash$ has a conservative model in some regular category (e.g.\ \cite[Proposition 6.4]{but98}).
\end{proof}

\addtocounter{theorem}{2}

\begin{theorem}
Let $\Scott{-}$ be a model of a theory $\vdash$ in $\L$.
Then the inchworm logic $\vdash_\C$ of $\vdash$ is sound with respect to the filter models in $\Scott{-}$:
If $\Gamma \vdash_\C \varphi$, then $G \vDash \varphi$ for every filter model $G$ of $\Gamma$ in $\Scott{-}$.
\end{theorem}

\begin{proof}
Assume the antecedent of \eqref{item:inchworm.base}, that there is $U \in \C$ such that $\varphi \in \Phi_U$ and $\Gamma_U \vdash \varphi$.
Fix any filter model $G$ of $\Gamma$ in $(\Scott{-}, F)$.
Then $\Scott{\psi}_U \in G_U$ for all $\psi \in \Gamma_U$, which implies $\bigwedge_{\psi \in \Gamma_U} \Scott{\psi}_U \in G_U$ because $G_U$ is a filter.
Yet $\Gamma_U \vdash \varphi$ means $\bigwedge_{\psi \in \Gamma_U} \Scott{\psi}_U \leqslant \Scott{\varphi}_U$.
Therefore $\Scott{\varphi}_U \in G_U$ since $G_U$ is closed upward.
Thus $G \vDash \varphi$.

Suppose that $G \vDash \varphi$ for every filter model $G$ of $\Gamma$ in $(\Scott{-}, F)$, and that $G \vDash \psi$ for every filter model $G$ of $\Delta$ in $(\Scott{-}, F)$.
Then every filter model $G$ of $\Gamma \cup \Delta$ in $(\Scott{-}, F)$ has $G \vDash \varphi$ and so $G \vDash \psi$.
Therefore, the induction with \eqref{item:inchworm.base} and \eqref{item:inchworm.inductive} shows the theorem.
\end{proof}

\begin{fact}
Let $(\Scott{-}, F)$ be a model of a theory $\vdash$ in $\L$.
Given any $\Gamma \subseteq \Phi_\C$, the family $\FiltMM[F]{\Gamma} = (\{\, S \monoto F_U \mid \Scott{\varphi}_U \leqslant S$ for some $\varphi \in {\Gamma^\ast}_U \,\})_{U \in \C}$ is a filter model of $\Gamma$ in $(\Scott{-}, F)$.
Moreover, for any filter model $G$ of $\Gamma$ in $(\Scott{-}, F)$, $\FiltMM[F]{\Gamma}_U \subseteq G_U$ for each $U \in \C$.
\end{fact}

\begin{proof}
Fix $U \in \C$.
$\FiltMM[F]{\Gamma}_U$ is upward closed by definition.
$\top \in \FiltMM[F]{\Gamma}_U$ since $\Gamma \vdash_\C \top$.
Now suppose $S, S' \in \FiltMM[F]{\Gamma}_U$, i.e.\ $\Scott{\varphi}_U \leqslant S$ and $\Scott{\psi}_U \leqslant S'$ for some $\varphi, \psi \in {\Gamma^\ast}_U$.
Then $\Scott{\varphi \wedge \psi}_U \leqslant S \wedge S'$, where $\varphi \wedge \psi \in {\Gamma^\ast}_U$;
so $S \wedge S' \in \FiltMM[F]{\Gamma}_U$.
Thus $\FiltMM[F]{\Gamma}_U$ is a filter.

Fixing $U \subseteq V \in \C$, we claim $\FiltMM[F]{\Gamma}_U = \{\, S \monoto F_U \mid {F_{U \subseteq V}}^{-1}(S) \in \FiltMM[F]{\Gamma}_V \,\}$.
Suppose $S \in \FiltMM[F]{\Gamma}_U$, i.e.\ $\Scott{\varphi}_U \leqslant S$ for some $\varphi \in {\Gamma^\ast}_U$.
Then $\Scott{\varphi}_V = {F_{U \subseteq V}}^{-1} \Scott{\varphi}_U \leqslant {F_{U \subseteq V}}^{-1}(S)$ for $\varphi \in {\Gamma^\ast}_U \subseteq {\Gamma^\ast}_V$.
Thus ${F_{U \subseteq V}}^{-1}(S) \in \FiltMM[F]{\Gamma}_V$.
On the other hand, suppose ${F_{U \subseteq V}}^{-1}(S) \in \FiltMM[F]{\Gamma}_V$, i.e.\ $\Scott{\psi}_V \leqslant {F_{U \subseteq V}}^{-1}(S)$ for some $\psi \in {\Gamma^\ast}_V$.
This implies by $\exists_{F_{U \subseteq V}} \dashv {F_{U \subseteq V}}^{-1}$ that $\Scott{\exists_{V \setminus U} \ldot \psi}_U \leqslant S$ for $\exists_{V \setminus U} \ldot \psi \in {\Gamma^\ast}_U$, i.e.\ $S \in \FiltMM[F]{\Gamma}_U$.
\end{proof}

\begin{lemma}
Suppose $\Scott{-}$ is a conservative model of a theory $\vdash$ in $\L$.
Then $\Gamma \vdash_\C \varphi$ iff $G \vDash \varphi$ for every filter model $G$ of $\Gamma$ in $\Scott{-}$.
\end{lemma}

\begin{proof}
Since the ``only if'' part is soundness (\autoref{thm:soundness.filter}), suppose $G \vDash \varphi$ for every filter model $G$ of $\Gamma$ in $\Scott{-}$.
Then $\FiltMM[F]{\Gamma} \vDash \varphi$ in particular.
This means that, for $U \in \C$ such that $\varphi \in \Phi_U$, $\Scott{\varphi}_U \in \FiltMM[F]{\Gamma}_U$, i.e.\ $\Scott{\psi}_U \leqslant \Scott{\varphi}_U$ for some $\psi \in {\Gamma^\ast}_U$.
Therefore $\Gamma \vdash_\C \psi$ and $\psi \vdash \varphi$ by the conservativity of $\Scott{-}$.
Thus $\Gamma \vdash_\C \varphi$.
\end{proof}

\end{document}